\documentclass[11pt]{article} 
\usepackage{fullpage,graphicx,amsmath,amsfonts,amssymb,amsthm,dsfont,authblk}
\usepackage[colorlinks]{hyperref}

\title{Quantum Algorithm for Simulating the Wave Equation}

\author[1]{Pedro C.S. Costa}
\author[2,3]{Stephen Jordan}
\author[3,4]{Aaron Ostrander}
\affil[1]{\small{Department of Physics and Astronomy, Macquarie University, Sydney, New South Wales 2109, Australia}}
\affil[2]{\small{Microsoft Quantum Architectures and Computation group, Redmond, WA}}
\affil[3]{\small{University of Maryland, College Park, MD}}
\affil[4]{\small{Joint Center for Quantum Information and Computer Science, College Park, MD}}

\date{}

\begin{document}

\bibliographystyle{unsrt}

\newcommand{\id}{\mathds{1}}                      
\newcommand{\ket}[1]{\left| #1\right\rangle}      
\newcommand{\bra}[1]{\left\langle #1\right|}      
\newcommand{\eq}[1]{(\ref{#1})}                   
\newcommand{\sect}[1]{\S\ref{#1}}                 
\renewcommand{\th}{^{\textrm{th}}}                    
\newtheorem{theorem}{Theorem}            

\maketitle

\begin{abstract}
We present a quantum algorithm for simulating the wave equation under Dirichlet and Neumann boundary conditions. The algorithm uses Hamiltonian simulation and quantum linear system algorithms as subroutines. It relies on factorizations of discretized Laplacian operators to allow for improved scaling in truncation errors and improved scaling for state preparation relative to general purpose linear differential equation algorithms. We also consider using Hamiltonian simulation for Klein-Gordon equations and Maxwell's equations.
\end{abstract}

\section{Introduction}\label{sec:intro}
Here we present a quantum algorithm for simulating the wave equation, subject to nontrivial boundary conditions. In particular, the algorithm can simulate the scattering of a wave packet off of scatterers of arbitrary shape, with either Dirichlet or Neumann boundary conditions. The output of the simulation is in the form of a quantum state proportional to the solution to the wave equation. By measuring this state one obtains a sample from a distribution proportional to the square of the amplitude, which in this case can be interpreted as the intensity of the wave.

Compared to classical algorithms, our method uses a number of qubits that scales only logarithmically with the number of lattice sites, whereas classical methods require a number of bits scaling linearly with the number of lattice sites. Additionally, for simulating the wave equation in a region of diameter $\ell$ in $D$-dimensions, discretized onto a lattice of spacing $a$, our quantum algorithm has a state-preparation step with time complexity $\widetilde{O}(D^{5/2} \ell/a)$ and a Hamiltonian simulation step with time complexity $ \widetilde{O} (TD^{2}/a)$, where $T$ is the evolution time for the wave equation. In contrast, all classical algorithms outputting a full description of the field, whether based on finite difference methods or finite element methods, must have time complexity scaling at least linearly with the number of lattice sites, \emph{i.e.} as $\Omega ((\ell/a)^D)$.

Several prior works give quantum algorithms for related problems. Berry gave an algorithm for first order linear differential equations that encodes a linear multistep method into a linear system which is then solved using a quantum linear system algorithm \cite{berry2014high}. This algorithm was recently improved upon in \cite{berry2017quantum} which gives an algorithm that scales better than the algorithm of \cite{berry2014high} with respect to several system parameters. Through standard transformations, the wave equation in a region of diameter $\ell$ can be discretized onto a lattice of spacing $a$ and transformed into a system of linear first order differential equations, which could then in general be solved by the quantum algorithms of \cite{berry2014high, berry2017quantum} with complexity of order $(\ell/a)^{2}$. (See \sect{sec:comparison}.) The complexity the quantum algorithm that we present here scales linearly with $(\ell/a)$. This quadratic improvement is achieved in exchange for being specialized for solving wave equations rather than general linear differential equations. At even greater generality, Leyton and Osborne proposed an algorithm for a class of nonlinear initial value problems \cite{leyton2008quantum}. This greater generality comes at a further cost in performance in that the complexity of the quantum algorithm scales exponentially with the evolution time. Related work on quantum algorithms for solving the Poisson equation can be found in \cite{cao2013quantum}.

The improved scaling of our algorithm relies on higher order approximations of the Laplacian operator and their factorizations using hypergraph incidence matrices. We describe how to find these operators and their hypergraph incidence matrices, and we provide numerical values for up to tenth order. (Throughout this manuscript we use the term $k\th$ order Laplacian to mean a discretization of the Laplacian which, when used on a lattice of spacing $a$, has leading error term of order $a^k$.) To our knowledge, these hypergraph incidence matrix factorizations do not appear elsewhere in the literature. These higher order Laplacians also allow us to improve how errors scale with respect to lattice spacing at the cost of simulating more complex (less sparse) Hamiltonians. { In particular, using a $s$-sparse Hamiltonian to simulate the wave equation for a volume of diameter $\ell$ in $D$ dimensions produces error on the order of $Ta^{2(s/D)-2}$, so $a$ scales as $(\epsilon / T)^{D/2(s-D)}$ (where $\epsilon$ is the error in the state output by the algorithm). Expressing the time complexity of our algorithm in terms of $\epsilon$ and $s$ , we find that the state preparation has time complexity 
	$\widetilde{O}(sD^{3/2} \ell (T / \epsilon)^{D/2(s-D)} )$
	and the Hamiltonian simulation has time complexity
	$\widetilde{O}(sDT (T / \epsilon)^{D/2(s-D)}  )$. Generally $s$ is an integer multiple of $D$, so these complexities scale polynomially in $D$ even though $D$ appears in an exponent.}

In \cite{clader2013preconditioned}, Jacobs, Clader, and Sprouse proposed a quantum algorithm for calculating electromagnetic scattering cross-sections that is based on solving boundary value problems in the special case of monochromatic waves. This monochromaticity assumption allows separation of variables thereby reducing the calculation to a time-independent problem. 

Rather than finite-difference methods, as discussed here, it is also possible to obtain approximate solutions to the full time-dependent wave equation through finite element methods such as the Galerkin method. In \cite{montanaro2016quantum} Montanaro and Pallister analyze, in a general context, the degree to which quantum linear algebra methods such as \cite{harrow2009quantum,childs2015quantum} allow speedup for finite element methods. Detailed analysis of how this can be applied to the wave equation specifically, particularly with the aid of preconditioners, is a complex subject which we defer to future work.

Following \cite{clader2013preconditioned} we consider as our primary application the simulation of scattering in complicated geometries \footnote{Note that, the presence of a scatterer breaks translational invariance and consequently the Laplacian cannot simply be diagonalized by a Fourier transform.}, as illustrated in figure \ref{scattering}. In this case, the initial condition at time zero is a localized wave packet and its time derivative, and the final output of the simulation algorithm is an estimate of the intensity of the wave at a later time $t$ within some region of space occuppied by the detector. After discretizing space, the scatterer can be modeled as a hole in the lattice where some points have been removed. Dirichlet or Neumann boundary conditions can be imposed on the boundary of this hole, as discussed in \sect{sec:bc}. In \sect{sec:init} we describe how to accommodate various initial conditions in our approach. In \sect{sec:numerics} we provide numerical evidence that our approach accurately simulates the wave equation with appropriate behavior at boundaries. In \sect{sec:highorder} and \sect{sec:higherbc} we describe higher order approximations of the Laplacian operator which allow for more precise approximations. In \sect{sec:errors} we provide numerical confirmation that higher order Laplacians improve how errors scale. In \sect{sec:postprocessing} we discuss the post-processing step which follows Hamiltonian simulation. In \sect{sec:comparison} we compare our approach to other quantum algorithms for the wave equation. In \sect{sec:KG} and \sect{sec:Maxwell} we address the use of Hamiltonian simulation for simulating the Klein-Gordon equation and Maxwell’s equations, respectively.

\begin{figure}
	\begin{center}
		\includegraphics[width=0.35\textwidth]{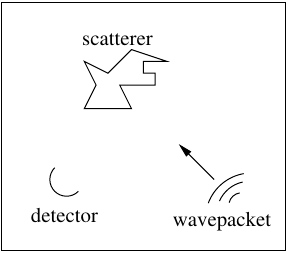}
	\end{center}
	\caption{\label{scattering} For a given initial wave packet and a given scatterer, we would like to estimate the resulting spatial distribution of wave intensity resulting at some later time $t$. In particular, one may wish to know the total intensity captured by a detector occuppying some region of space. This can be estimated using a quantum simulation in which the wavefunction directly mimics the dynamics of the solution to the wave equation. The final intensity in the detector region is equal to the probability associated with the corresponding part of the Hilbert space, which can be estimated from the statistics resulting from a projective measurement.}
\end{figure}

\section{Algorithm}
In any number of dimensions, the wave equation is
\begin{equation}
\label{waveq}
\frac{d^2}{dt^2} \phi = c^2 \nabla^2 \phi.
\end{equation}
To avoid cumbersome notation, in the rest of this paper we will take the wave propagation speed to be $c=1$. For a given initial condition specifying $\phi(\vec{x},t)$ and $\frac{d \phi (\vec{x},t)}{dt}$ at $t=0$, our goal is to obtain a quantum state encoding the solution $\phi(\vec x, T)$ determined by \eq{waveq} at some later time $t$. 

To achieve this, we will first discretize space. We can then think of $\nabla^2$ as a matrix acting on a vector $\phi$ whose entries encode the value of the field at each point in discrete space (with appropriate boundary conditions). Discrete approximations of the Laplacian operator have been thoroughly studied in both spectral graph theory and quantum chemistry, and we draw upon this previous work. In the simplest case, we can discretize a finite region of $\mathbb{R}^n$ onto a cubic grid of lattice spacing $a$. The resulting points can be thought of as a graph $G_a$, with edges between nearest neighbors. The corresponding graph Laplacian $L(G_a)$ is the square matrix whose rows and columns index the vertices of this graph, and whose off-diagonal matrix elements are minus one for connected vertices and zero otherwise. Each diagonal matrix element is equal to the degree of the corresponding vertex, \emph{i.e.} the number of other vertices it is connected to. The operator $-\frac{1}{a^2} L(G_a)$ approximates $\nabla^2$ in the limit $a \to 0$. For example, in one dimension:
\begin{equation}
- \frac{1}{a^2} \left[ L(G_a) \phi \right]_j = \frac{\phi_{j-1} - 2 \phi_j + \phi_{j+1}}{a^2},
\end{equation}
which becomes the second derivative of $\phi$ in the limit $a \to 0$. At finite $a$ the truncation error is $O(a^2)$.

After discretization, we are faced with the task of simulating
\begin{equation}
\frac{d^2}{dt^2} \phi = - \frac{1}{a^2} L \phi.
\end{equation}
To this end, consider a Hamiltonian of the following block form, which by construction is Hermitian independent of the specific choice of matrix $B$.
\begin{equation}
\label{Block_hamiltonian}
H=\frac{1}{a}\begin{bmatrix}0 & B\\
B^\dag & 0
\end{bmatrix}.
\end{equation}
Schr\"{o}dinger's equation then takes the form
\begin{align}
\label{Wave_via_Schr}
\frac{d}{dt} \begin{bmatrix} \phi_V \\ \phi_E \end{bmatrix} & =  \frac{-i}{a}\begin{bmatrix}0 & B\\
B^\dag & 0
\end{bmatrix} \begin{bmatrix} \phi_V \\ \phi_E \end{bmatrix}
\end{align}
which implies
\begin{align}
\frac{d^2}{dt^2} \begin{bmatrix} \phi_V \\ \phi_E \end{bmatrix} & =  \frac{-1}{a^2}\begin{bmatrix}0 & B\\
B^\dag & 0
\end{bmatrix}^2 \begin{bmatrix} \phi_V \\ \phi_E \end{bmatrix} \\
& =  \frac{-1}{a^2}\begin{bmatrix} B B^\dag & 0\\
0 & B^\dag B
\end{bmatrix} \begin{bmatrix} \phi_V \\ \phi_E \end{bmatrix} \label{decouple}
\end{align}
So, if $B B^\dag = L$ then a subspace of the full Hilbert space evolves according to a discretized wave equation.

For any graph, weighted or unweighted, and with or without self-loops, $B B^T = L$ is achieved by taking $B$ to be the corresponding signed incidence matrix, defined as follows. For a given graph with $|V|$ vertices and $|E|$ edges, $B$ is an $|V| \times |E|$ matrix with rows indexed by vertices and columns indexed by edges. One starts by arbitrarily assigning orientations to the edges of the graph. This arbitrary choice affects $B$ but does not affect $B B^T$, which always equals the Laplacian of the undirected graph. The general definition of the incidence matrix for a graph where edge $j$ has weight $W_j$ is
\begin{equation}
\label{incidence_definition}
B_{ij}=\begin{cases}
\sqrt{W_{j}} & \text{if \ensuremath{j} is a self-loop of \ensuremath{i}}\\
\sqrt{W_{j}} & \text{if \ensuremath{j} is an edge with \ensuremath{i,} as source,}\\
-\sqrt{W_{j}} & \text{\text{if \ensuremath{j} is an edge with \ensuremath{i} as sink,}}\\
0 & \text{otherwise.}
\end{cases}
\end{equation}
In the special case that the graph is unweighted, $W_j=1$ for every edge.

From the above, one sees that the Hilbert space associated with the graph is
\begin{equation}
\mathcal{H}=\mathcal{H}_{V} \oplus \mathcal{H}_{E},
\end{equation}
where $\mathcal{H}_V$ is the vertex space (where $\phi_V$ is supported) and $\mathcal{H}_{E}$ is the edge space (where $\phi_E$ is supported). The dynamics on the vertex space obeys the discretized wave equation. The amplitudes associated with the edges are extra variables that necessarily arise when converting second order differential equations into first order differential equations.

Simulating the time evolution according to \eq{Wave_via_Schr} can be achieved using state of the art quantum algorithms for simulating the dynamics induced by general sparse Hamiltonians. One sees that the dimension of the Hilbert space $\mathcal{H}$ is equal to the number of vertices of the graph plus the number of edges: $|V|$ + $|E|$. In particular, for a cubic region of side-length $l$ in $D$-dimensions, discretized into a cubic grid of lattice spacing $a$, one has $|V| = (l/a)^D$ and $|E| = D (l/a)^D$. Thus, the number of qubits needed is $\log_2 \left[ (1+D) (l/a)^D \right]$. The largest matrix element of $H$ has magnitude $1/a$, and the number of nonzero matrix elements in each row or column of $H$ is at most $2D$.  

Using the method of \cite{berry2015hamiltonian} we can approximate the unitary time evolution $e^{-iHt}$ to within $\epsilon$ using a quantum circuit of
\begin{equation} g=O\left[\tau\left[n+\log^{5/2}\left(\tau/\epsilon\right)\right]\frac{\log\left(\tau/\epsilon\right)}{\log\log\left(\tau/\epsilon\right)}\right],
\end{equation}
gates, where $\tau=s\left\Vert H\right\Vert _{\text{max}}t$, where $\left\Vert H\right\Vert _{\text{max}}$ is the largest matrix element of $H$ in absolute value, $s=\text{sparsity of \ensuremath{H}}$ and $n=\text{number of qubits}$. For the Hamiltonian of \eq{Block_hamiltonian}, $s=2D$, $\left\Vert H\right\Vert _{\text{max}} = 1/a$, and $n = \log_2 \left[ (1+D) (l/a)^D \right]$, and therefore the total complexity of simulating the time-evolution is
\begin{eqnarray}
g & = & O \left[ \frac{Dt}{a} \left( \log \left( (1+D) (l/a)^D \right) + \log^{5/2} \left( \frac{2Dt}{a \epsilon} \right) \right) \frac{\log \left( \frac{2Dt}{a \epsilon} \right)}{\log \log \left( \frac{2Dt}{a \epsilon} \right)} \right] \nonumber \\
& = & \widetilde{O} \left[ \frac{tD^2}{a} \right],
\end{eqnarray}
where the notation $\widetilde{O}$ indicates that we are suppressing logarithmic factors. The table below compares the asymptotic runtime and memory usage of our algorithm against standard classical numerical methods for solving differential equations.
\noindent \begin{center}
	\begin{tabular}{|l|c|c|}
		\hline
		& Classical & Quantum\tabularnewline
		\hline
		\hline
		Time & $\Omega\left[T(l/a)^h\right]$ & $\widetilde{O} \left[ tD^2/a \right]$ \tabularnewline
		\hline
		Space & $(l/a)^h$ &$D\log(l/a)$\tabularnewline
		\hline
	\end{tabular}
	\par\end{center}

The remaining considerations are the implementation of desired boundary conditions, the preparation of an initial state implementing the desired initial conditions, errors induced by discretizing the wave equation, and the relative probability to obtain samples from the vertex space versus the edge space at the end of the computation. In the following sections we address each of these issues in turn. These considerations motivate various improvements and extensions to the above algorithm, which we introduce along the way, in particular the use of higher order discretizations of $\nabla^2$. 


\section{Boundary Conditions}
\label{sec:bc}

Here we will consider how to implement two commonly used boundary conditions: Dirichlet and Neumann. With Dirichlet boundary conditions $\phi = 0$ at the boundary. With Neumann boundary conditions $\nabla \phi \cdot \hat{n} = 0$ at the boundary, where $\hat{n}$ is the unit vector normal to the boundary. For any shape of boundary and in any number of dimensions our prescription is as follows. To implement Neumann boundary conditions use the ordinary graph Laplacian of the graph obtained by starting with the cubic grid and removing the vertices interior to the scattering object. To implement Dirichlet boundary conditions one must add weighted self-loops to each of the vertices on the boundary with weights equal to the number of edges that are missing relative to interior vertices. (This ensures that the diagonal matrix elements of the resulting graph Laplacian are all equal.) See figure \ref{bc} for an illustration. For pedagogical reasons, we give two derivations of the Laplacians implementing these boundary conditions, using the one dimensional path graph as an example. One derivation is based on discretization of derivatives, and the other is by linear algebra on an already-discretized system.
\begin{figure}
	\includegraphics[trim=-180 50 50 50,clip,scale=0.6]{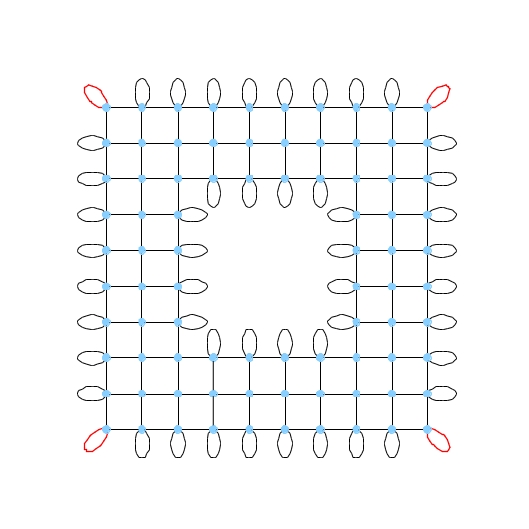}\\
	\caption{\label{bc}To implement Dirichlet boundary conditions in a discretize square region with a square hole, one adds self-loops as illustrated above. The thick red self loops at the corners have weight two. All other edges (self-loops and otherwise) have weight one. This prescription was used in the numerical examples of \sect{sec:numerics}. To implement Neumann boundary conditions one omits all self-loops.}
\end{figure}

\subsection*{Neumann Boundary Conditions by Discretization}

Consider the line segment $[0,1]$. Within this, the second derivative discretizes to
\begin{equation}
\frac{d^2 \phi}{d x^2} = \lim_{a \to 0} \frac{\frac{d \phi}{dx} (x+a/2) - \frac{d \phi}{dx}(x-a/2)}{a}  
= \lim_{a \to 0} \frac{\frac{\phi(x+a)-\phi(x)}{a} - \frac{\phi(x)-\phi(x-a)}{a}}{a}.  
\end{equation}
This yields at internal vertices the familiar form of a discrete Laplacian.
\begin{equation}
\frac{d^2 \phi}{d x^2}(x) = \lim_{a \to 0} \frac{\phi(x+a) - 2 \phi(x) + \phi(x-a)}{a^2}
\end{equation}
With Neumann boundary conditions, $\frac{d \phi}{dx} = 0$ at the boundaries. Thus, at the leftmost vertex we have:
\begin{eqnarray}
\frac{d^2 \phi}{d x^2}(0) &=& \lim_{a \to 0} \frac{\frac{d \phi}{dx} (a/2) - \frac{d \phi}{dx}(-a/2)}{a}\\ 
&=& \lim_{a \to 0} \frac{\frac{d \phi}{dx} (a/2)}{a} \nonumber \\ 
&=& \lim_{a \to 0} \frac{\phi(a)-\phi(0)}{a^2}. \nonumber
\end{eqnarray}

Similarly, $\frac{d \phi}{dx}(x+a/2)$ vanishes at the rightmost vertex. For example, if we discretize the segment $[0,1]$ into five lattice sites we would have
\begin{equation}
-\frac{1}{a^2} L_{\mathrm{Neumann}} \phi = \frac{1}{a^2} \left[ \begin{array}{rrrrr}
-1 & 1 & 0 & 0 & 0 \\
1 & -2 & 1 & 0 & 0 \\
0 & 1 & -2 & 1 & 0 \\
0 & 0 & 1 & -2 & 1 \\
0 & 0 & 0 & 1 & -1
\end{array} \right]
\left[ \begin{array}{c}
\phi(0) \\
\phi(a)\\
\phi(2a) \\
\phi(3a) \\
\phi(4a)
\end{array} \right].
\end{equation}
$L_{\mathrm{Neumann}}$ is recognizable as the ordinary graph Laplacian for the path graph of five vertices:
\begin{center}
	\includegraphics[width=2in]{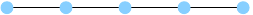}.
\end{center}
This holds more generally; the ordinary graph Laplacian on discretized regions of any shape in any number of dimensions yields Neumann boundary conditions. Note that in the above example discretizing the unit interval with five vertices, one should take $a=1/5$ because each of the four edges in the graph corresponds to a distance of $a$, but as we see from the above argument, the boundary conditions correspond to $d \phi/dx = 0$ at $x=-a/2$ and $x = 1+a/2$.

\subsection*{Neumann Boundary Conditions, Algebraic Derivation}
We first consider the Laplacian $L$ for an infinite path graph with vertices labeled by $\mathbb{Z}$, which is a tridiagonal matrix with 2 on the diagonal and -1 on the off-diagonals.
It suffices to consider imposing the boundary conditions at the left end of the interval, which we assume corresponds to the vertex 0 in our discrete space. Then for Neumann boundary conditions the field $\phi$ is constant on all vertices $v \in \mathbb{Z}^-$, that is $\phi_v= \phi_0$. Then consider how $L$ acts on the field in the neighborhood of 0. We represent this as

\begin{align}
L \vec \phi &= \begin{bmatrix}
2 & -1 & 0 & 0 & 0\\
-1 & 2 & -1 & 0 & 0\\
0 & -1 & 2 & -1 & 0\\
0 & 0 & -1 & 2 & -1 \\
0 & 0 & 0 & -1 & 2 
\end{bmatrix}
\begin{bmatrix}
\phi_{-2} \\
\phi_{-1} \\
\phi_{0} \\
\phi_{1} \\
\phi_{2}
\end{bmatrix}\\
\mapsto L_{Neumann} \vec \phi &=  \begin{bmatrix}
2 & -1 & 0 & 0 & 0\\
-1 & 2 & -1 & 0 & 0\\
0 & -1 & 2 & -1 & 0\\
0 & 0 & -1 & 2 & -1 \\
0 & 0 & 0 & -1 & 2 
\end{bmatrix}
\begin{bmatrix}
\phi_{0} \\
\phi_{0} \\
\phi_{0} \\
\phi_{1} \\
\phi_{2}
\end{bmatrix}
=  \begin{bmatrix}
0 \\
0 \\
\phi_0 - \phi_1 \\
2 \phi_1 - \phi_0 - \phi_2 \\
\dots
\end{bmatrix}\\
& =  \begin{bmatrix}
0 & 0 & 0 & 0 & 0\\
0 & 0 & 0 & 0 & 0\\
0 & 0 & 1 & -1 & 0\\
0 & 0 & -1 & 2 & -1 \\
0 & 0 & 0 & -1 & 2 
\end{bmatrix}
\begin{bmatrix}
\phi_{0} \\
\phi_{0} \\
\phi_{0} \\
\phi_{1} \\
\phi_{2}
\end{bmatrix}
\end{align}

So we see that imposing Neumann boundary conditions allows us to ignore the vertices labeled by negative numbers. To give a finite example, if we restrict to the vertices $0,1,2,3$ (i.e. impose Neumann boundary conditions for vertices to the left of 0 and to the right of 3) then the Laplacian we produce is 

\begin{equation} \label{Laplacian_N}
L=
\begin{bmatrix}
1 & -1 & 0 & 0\\
-1 & 2 & -1 & 0 \\
0 & -1 & 2 & -1 \\
0 & 0 & -1 & 1
\end{bmatrix},
\end{equation}

which is exactly the graph Laplacian for the path graph on 4 vertices.

\subsection*{Dirichlet Boundary Conditions, Algebraic Derivation}

We use similar arguments to show how to impose Dirichlet boundary conditions. Consider imposing $\phi =0$ to the left of $0$. Then $L$ acts as

\begin{align}
L \vec \phi &= \begin{bmatrix}
2 & -1 & 0 & 0 & 0\\
-1 & 2 & -1 & 0 & 0\\
0 & -1 & 2 & -1 & 0\\
0 & 0 & -1 & 2 & -1 \\
0 & 0 & 0 & -1 & 2 
\end{bmatrix}
\begin{bmatrix}
\phi_{-2} \\
\phi_{-1} \\
\phi_{0} \\
\phi_{1} \\
\phi_{2}
\end{bmatrix}\\
\mapsto L_{Dirichlet} \vec \phi &=  \begin{bmatrix}
2 & -1 & 0 & 0 & 0\\
-1 & 2 & -1 & 0 & 0\\
0 & -1 & 2 & -1 & 0\\
0 & 0 & -1 & 2 & -1 \\
0 & 0 & 0 & -1 & 2 
\end{bmatrix}
\begin{bmatrix}
0 \\
0 \\
\phi_{0} \\
\phi_{1} \\
\phi_{2}
\end{bmatrix}
=  \begin{bmatrix}
0 \\
- \phi_0 \\
2 \phi_0 - \phi_1 \\
2 \phi_1 - \phi_0 - \phi_2 \\
\dots
\end{bmatrix}\\
& =  \begin{bmatrix}
0 & 0 & 0 & 0 & 0\\
0 & 0 & -1 & 0 & 0\\
0 & 0 & 2 & -1 & 0\\
0 & 0 & -1 & 2 & -1 \\
0 & 0 & 0 & -1 & 2 
\end{bmatrix}
\begin{bmatrix}
0 \\
0 \\
\phi_{0} \\
\phi_{1} \\
\phi_{2}
\end{bmatrix} \label{2newLaplacian}
\end{align}

Since we are only concerned with how the Laplacian acts on vertices $0,1,2 \dots$ and not on $-1$ we can ignore the fact that $(L \vec \phi )_{-1} = - \phi_0$. Another way to motivate this is that by restricting the wave equation to act on vertices $0,1,2 \dots$ we do not provide a dynamical equation for $\phi_{-1}$, so it will remain 0.

To compare this with the Neumann case, if we restrict to the vertices $0,1,2,3$ then the Laplacian we produce is 

\begin{equation}
\label{Laplacian_D}L=
\begin{bmatrix}
2 & -1 & 0 & 0\\
-1 & 2 & -1 & 0 \\
0 & -1 & 2 & -1 \\
0 & 0 & -1 & 2
\end{bmatrix},
\end{equation}
which differs from the Neumann Laplacian in the upper-left and lower-right entries.

\section{Initial Conditions}
\label{sec:init}

The first step in our quantum algorithm is to prepare a quantum state $[\phi_V, \phi_E]$ corresponding to desired initial conditions $\phi(x)$ and $\frac{{\partial} \phi}{{\partial} t}(x)$ at $t=0$. Our method for preparing the initial state and its complexity varies depending on the specific type of initial conditions.

As a first example, consider a line-segment with Dirichlet boundary conditions, discretized into four lattice sites. In this case, by \eq{Block_hamiltonian} and \eq{incidence_definition}, we have
\begin{equation}
H = \frac{1}{a} \left[ \begin{array}{rrrrrrrrr}
0 & 0 & 0 & 0 & 1 & 1 & 0 & 0 & 0\\
0 & 0 & 0 & 0 & 0 & -1 & 1 & 0 & 0\\
0 & 0 & 0 & 0 & 0 & 0 & -1 & 1 & 0\\
0 & 0 & 0 & 0 & 0 & 0 & 0 & -1 & 1\\
1 & 0 & 0 & 0 & 0 & 0 & 0 & 0 & 0\\
1 & -1 & 0 & 0 & 0 & 0 & 0 & 0 & 0\\
0 & 1 & -1 & 0 & 0 & 0 & 0 & 0 & 0\\
0 & 0 & 1 & -1 & 0 & 0 & 0 & 0 & 0\\
0 & 0 & 0 & 1 & 0 & 0 & 0 & 0 & 0
\end{array}
\right].
\end{equation}
This can be viewed as a discretization of
\begin{equation}
H = \left[ \begin{array}{cc}
0 & \frac{d}{dx} \\
-\frac{d}{dx} & 0
\end{array} \right] \label{1dcontinuum}
\end{equation}
where we use the forward difference to approximate $\frac{d}{dx}$ and the backward difference to approximate $-\frac{d}{dx}$. More generally, in an arbitrary number of dimensions, the Hamiltonian \eq{Block_hamiltonian} can be seen as a discretization of
\begin{equation}
H = \left[ \begin{array}{cc}
0 & \vec{\nabla}^T \\
-\vec{\nabla} & 0
\end{array} \right].
\end{equation}
(We here view $\phi_E$ as describing a vector field, where the value associated with a given edge in the graph is the vector component along the direction that the edge points.) Consequently, for an arbitrary initial condition specified by $\phi_0(x)$ and $\frac{d}{dt} \phi_0(x)$ one must prepare a corresponding initial quantum state that is a solution to
\begin{eqnarray}
\phi_V & = & \phi_0 \nonumber \\
\vec{\nabla} \cdot \vec{\phi}_E & = & i \frac{d}{dt} \phi_0. \label{Gausslike}
\end{eqnarray}
In more than one dimension, the equation \eq{Gausslike} does not uniquely determine $\phi_E$ since $\vec{\nabla} \times \vec{\phi}_E$ is unspecified. (In one dimension $\phi_E$ is determined up to an additive constant.) In the remainder of this section we consider how to compute a solution to \eq{Gausslike} and how to prepare the initial state $[\phi_V, \phi_E]$ on a quantum computer in various cases of interest.
\subsection{Static Initial State}

The simplest case is to prepare a state with $\frac{d}{dt} \phi$ uniformly equal to zero. Then, one can use $\phi_E = 0$ as an initial quantum state. The state preparation problem then reduces to preparing $\phi_V$; however, this is not necessarily efficient for arbitrary $\phi_V$. Preparation of a completely arbitrary quantum state in an $N$-dimensional Hilbert space has complexity of order $N$, \emph{i.e.} exponential in the number of qubits. Specifically, suppose one were given an oracle, which when queried with a bit string $x$ returned a corresponding amplitude $\psi(x)$ written (to some number of bits of precision) into an output register. One wishes to prepare the corresponding quantum state $\ket{\psi} = \sum_{x \in \{0,1\}^n} \psi(x) \ket{x}$. The worst-case complexity of this task is $\Theta(\sqrt{N})$~\cite{Grover2000}. In many cases of interest, the complexity for preparing the initial state may be much lower. 
In particular, as was originally shown in \cite{Zalka}, a state of the form
\begin{equation}
\sum_{x \in \{0,1\}^n} \sqrt{p(x)} \ket{x}
\end{equation}
can be prepared in poly($n$) time on a quantum computer provided that each of the conditional probabilities
\begin{equation}
p(x_1 x_2 \ldots x_r | x_{r+1} x_{r+2} \ldots x_n) \quad \quad r = 1,2,3, \ldots 
\end{equation}
can be efficiently computed. As discussed in \cite{Grover_Rudolph}, these conditional probabilities can be efficiently computed for all log-concave probability distributions.

\subsection{Rigidly Translating Wave packet}

In one spatial dimension, for any twice-differentiable wave packet shape $w$,
\begin{equation}
\phi(\vec{x},t) = w(x - ct)
\end{equation}
is a solution to the wave equation $\frac{d^2}{dt^2} \phi = c^2 \nabla^2 \phi$. (In this manuscript we will generally take $c=1$.) From \eq{1dcontinuum} one sees that the quantum state
\begin{equation}
\label{rigid1d}
\left[ \begin{array}{c} \phi_V \\ \phi_E \end{array} \right] = \left[ \begin{array}{c} w(x - t) \\ i w(x-t) \end{array} \right]
\end{equation}
represents this solution in the continuum limit. For a lattice with Neumann or Dirichlet boundary conditions, the vertex and edge Hilbert spaces have different dimensions, so the initial state is not merely $(|0 \rangle + i | 1 \rangle) | w(0) \rangle / \sqrt{2}$ where $| w \rangle \propto \sum_{x} w(x) \ket{x}$. This can be overcome by instead preparing $(|0 \rangle | w(0)_V \rangle + i | 1 \rangle | w(0)_E \rangle) / \sqrt{2}$ where $| w(0)_V \rangle \propto \sum_{j \in V} w(ja) \ket{j}$ and $| w(0)_E \rangle \propto \sum_{(j,k) \in E} w((j+k)a/2) \ket{(j,k)}$. So if the quantum state $\sum_{x} w(x) \ket{x}$ (suitably discretized in each Hilbert space) can be prepared in polynomial time then so can the state \eq{rigid1d}. More generally, in an arbitrary number of dimensions, one can obtain an analogous initial state proportional to
\begin{equation}
\left[ \begin{array}{c} \phi_V \\ \vec{\phi}_E \end{array} \right] = 
\left[ \begin{array}{c} w(x) \\ i \vec{v} w(x) \end{array} \right]
\end{equation}
with $|\vec{v}| = c$. This initially represents a wave packet traveling with velocity $\vec{v}$, but unlike in the one dimensional case, the wave packet will evetually suffer dispersion rather than simply translating rigidly.

\subsection{General Case}
\label{sec:general}

In the general case we may imagine that we are given efficient quantum circuits preparing the states $\ket{\phi_0} = \sum_{\vec{x}} \phi(\vec{x},0) \ket{\vec{x}}$ and $\ket{\dot{\phi}_0} \equiv \sum_{\vec{x}} {\left.\frac{\partial\phi\left(x,t\right)}{\partial t}\right|_{t=0}} \ket{\vec{x}}$.
The discrete analogue of \eq{Gausslike} is, via our incidence-matrix discretization:
\begin{eqnarray}
\phi_V & = & \phi_0 \\
-\frac{i}{a} B \phi_E & = & \dot{\phi}_0.
\end{eqnarray}
In two and higher dimensions, the solution to $\frac{i}{a} B \phi_E = \dot{\phi}_0$ is non-unique in general since the number of edges in the graph $G_a$ exceeds the number of vertices. Thus, the number of columns of $B$ exceeds the number of rows by a factor of order $D$, the number of spatial dimensions. One valid solution is to use as our quantum initial state
\begin{equation}
\label{init}
\left[ \begin{array}{c} \phi_V \\ \phi_E \end{array} \right]  \propto \left[ \begin{array}{c} 
\phi_0 \\ i a B^+ \dot{\phi}_0 \end{array} \right]
\end{equation}
where $B^+$ denotes the Moore-Penrose pseudoinverse of the matrix $B$.  A Moore-Penrose pseudoinverse has the property that the image of $B^+$ is the orthogonal complement of the kernel of $B$. Recall that $B$ is a map from $H_E \to H_V$. For the case of the standard 2nd order Laplacian, the corresponding $B$ is the signed incidence matrix of a graph. In this case $B$ can be interpreted in the continuum limit as a divergence. The Helmholtz decomposition theorem says that any twice-differentiable vector field can be decomposed into a curl-free component and a divergence-free component, the latter of which corresponds to the kernel of $B$ in the continuum limit. Thus, $\phi_E = -i a B^+ \dot{\phi}_0$ corresponds in the continuum limit to the solution to the following system of equations.
\begin{eqnarray}
\vec{\nabla} \cdot \vec{\phi}_E & = & -i \dot{\phi}_0 \\
\vec{\nabla} \times \vec{\phi}_E & = & 0.
\end{eqnarray}

To construct the state \eq{init} we can use the quantum linear systems algorithm of \cite{childs2015quantum}. Specifically, we wish to prepare the state proportional to the solution to $A x = b$ where
\begin{eqnarray}
A & = & \left[ \begin{array}{cc} \id & 0 \\ 0 & i a^{-1} B \end{array} \right] \\
b & = & \left[ \begin{array}{c} \phi_0 \\ \dot{\phi}_0 \end{array} \right]
\end{eqnarray}
This can be done using the quantum linear systems algorithm of \cite{childs2015quantum}, whose time complexity is $\widetilde{O}(\kappa)$, where $\kappa$ is the condition number of $A$, which in this case is equal to the condition number of the incidence matrix $B$.

Using the state \eq{init} restricts the classes of solutions which our algorithm simulates. 
This is because $B^+ \dot{\phi_0}$ does not have support in the kernel of $B$.
This is significant for Neumann boundary conditions because the kernel of $B$ (and of the Laplacian) is the all-ones vector, whereas for Dirichlet boundary conditions the kernel is trivial.
This means that, even if $ \dot{\phi_0}$ had support in the space spanned by the all-ones vector, the algorithm will simulate the system with the modified initial condition where $ \dot{\phi_0}$ does not have support in this space. This restriction may seem artificial, but it is a natural consequence of the unitarity of Hamiltonian dynamics. If the uniform support of $ \dot{\phi_0}$ were not projected out, then our algorithm would be able to simulate the solution $\phi (\vec x,t) \propto t$ (with no dependence on $\vec x$) for which $ \dot{\phi}$ is constant. This would result in the norm of the quantum state changing in time, in contradiction to unitarity.

In more detail, an algorithm from \cite{childs2015quantum} can perform the transformation $\dot{\phi}_0 \to -i a B^+ \dot{\phi}_0$ using a number of gates that scales as $\widetilde{O}(s \kappa \log N)$ where $s$ is the sparsity of $B$, $\kappa$ is the condition number of $B$ and $N$ is the dimension of the Hilbert space. The condition number of $B$ is the square root of the condition number of the graph Laplacian $L$. $L$ has norm $O(D)$ and smallest eigenvalue $O(\ell^2/a^2)$, independent of $D$, where the volume under consideration is $\ell \times \ell \times \ldots \times \ell$ which discretized onto a grid of spacing $a$. Thus $\kappa \sim \sqrt{D} \ell/a$. The sparsity of $B$ is $s \sim D$ for any fixed order of discretization, and the Hilbert space dimension is $N \sim (\ell/a)^D$. Putting this together yields an overall complexity of $\widetilde{O}(D^{5/2} \ell a^{-1})$ for state preparation in this case (neglecting log factors).



\section{Numerical Examples}
\label{sec:numerics}

The above analysis can be confirmed by numerical examples, as shown in this section. In all cases one sees that the dynamics and implementation of initial conditions and boundary conditions are consistent with theoretical expectations. Our quantum algorithm is implemented on a gate model quantum computer, and time evolution is discretized into a sequence of elementary gates through the method of \cite{berry2015hamiltonian}. The error induced by this time discretization is rigously upper bounded in \cite{berry2015hamiltonian}. Thus the focus of our numerical study is to investigate the errors induced by spatial discretization and verify the implementation of boundary conditions and initial conditions. To this end, we use the Dormand-Prince method\footnote{This is implemented as ODE45 in MATLAB.} \cite{DP80} (a variant of Runge-Kutta) to solve Schr{\"o}dinger's equation with Hamiltonians arising from our incidence matrix prescription.

As we know from \cite{STAB} there is a relation between the timestep and the lattice spacing that is necessary, but not sufficient, to keep the numerical simulations stable, which is
\[
\Delta t < a,
\] 
because of that we used this relation in all our numerical analyses.
In small examples we verified the accuracy numerical solution to the differential equations by comparing against direct computation of the entire unitary operator $e^{-iHt}$ applied to the initial state vector.

\begin{figure}[ht]
	\begin{tabular}{cc}
		\includegraphics[scale=0.15]{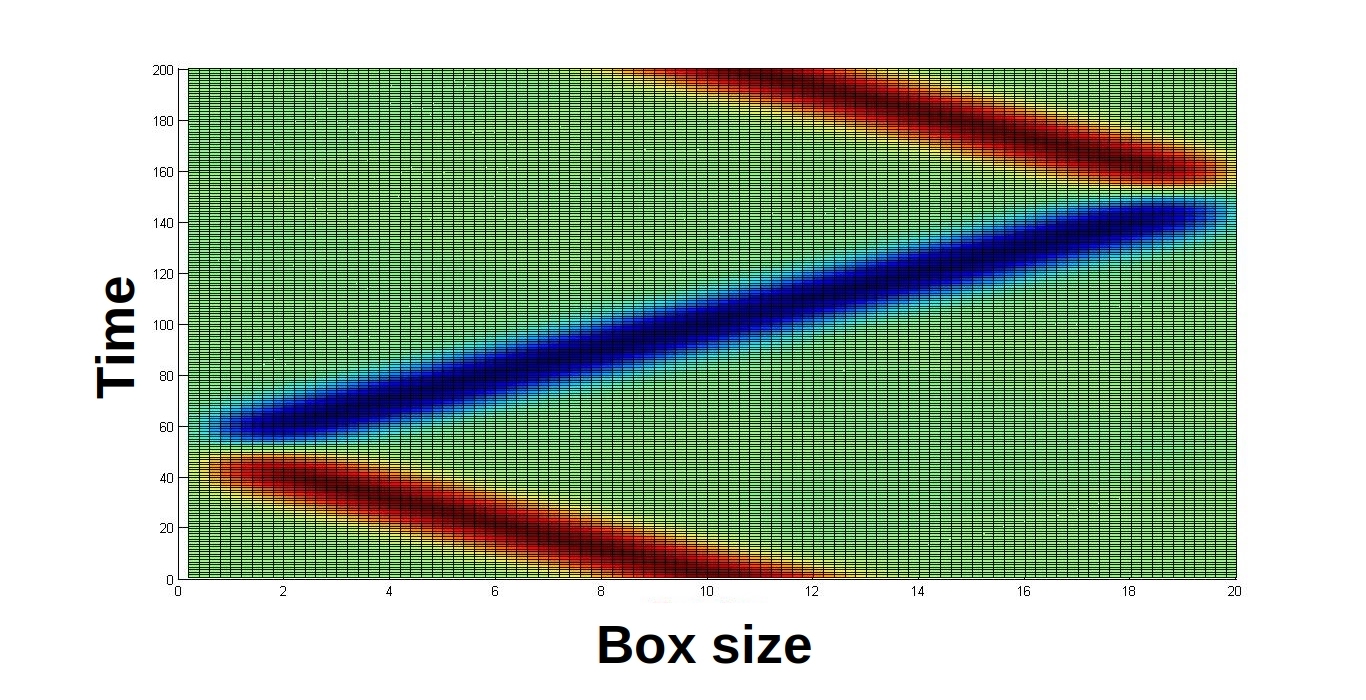}
		&\includegraphics[scale=0.18]{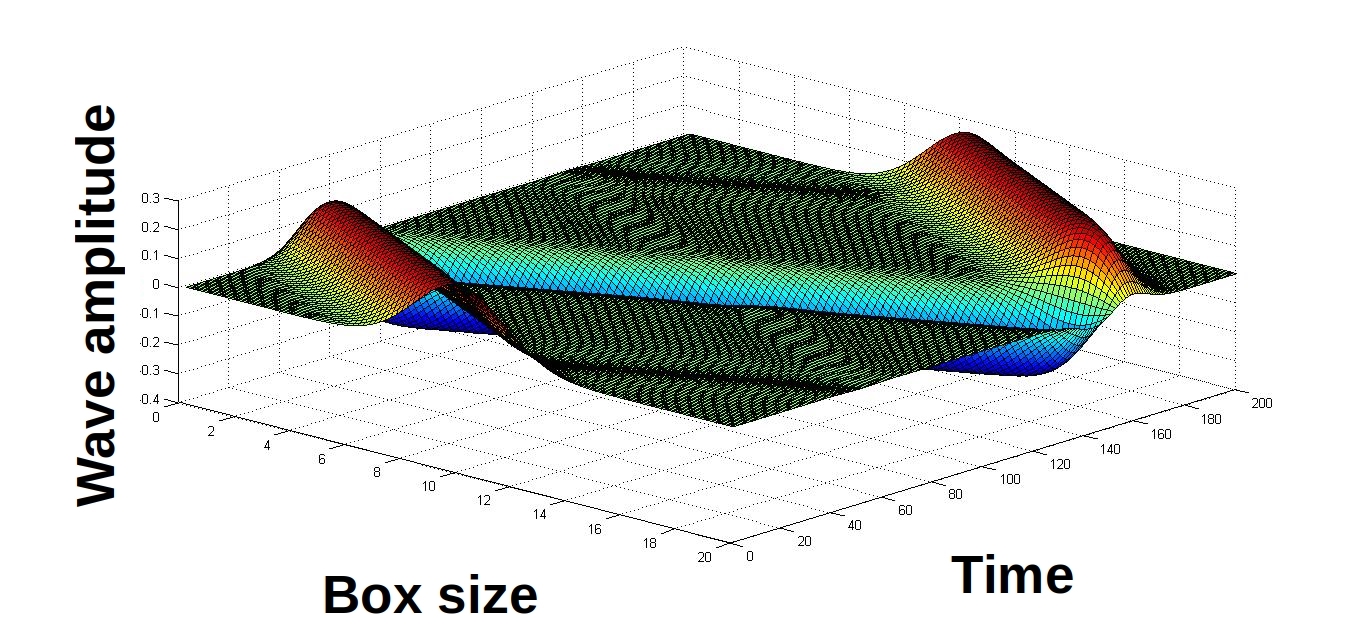}\\
		a)
		&
		b)
	\end{tabular}
	\caption{{\small \label{fig:shape_preser_one} \textbf{Shape preserving on line segment Dirichlet.} Here we consider the case of a rigidly-translating wave packet as described by \eq{rigid1d}. We can see two different views of the same wave packet starting in the middle point in a box with size 20, where space is represented by the $x$-axis while in the $y$-axis we have the time and the units are meters and seconds respectively. We can see the packet going back and forward between the extremes of the box. Although its wave amplitude is preserved in time, when the wave packet arrives at the end points the amplitude reflects simultaneously with the propagation's direction. The red color gives us the positive amplitude against the blue one with negative value. In figure $b$ the amplitude height is plotted in the $z$=axis and its units are meters. In this example we choose lattice spacing $a= 0.2469$ and gaussian wave packet of width $\sigma=1.6$.}}
\end{figure}

\begin{figure}[ht]
	\begin{tabular}{cc}
		\includegraphics[scale=0.15]{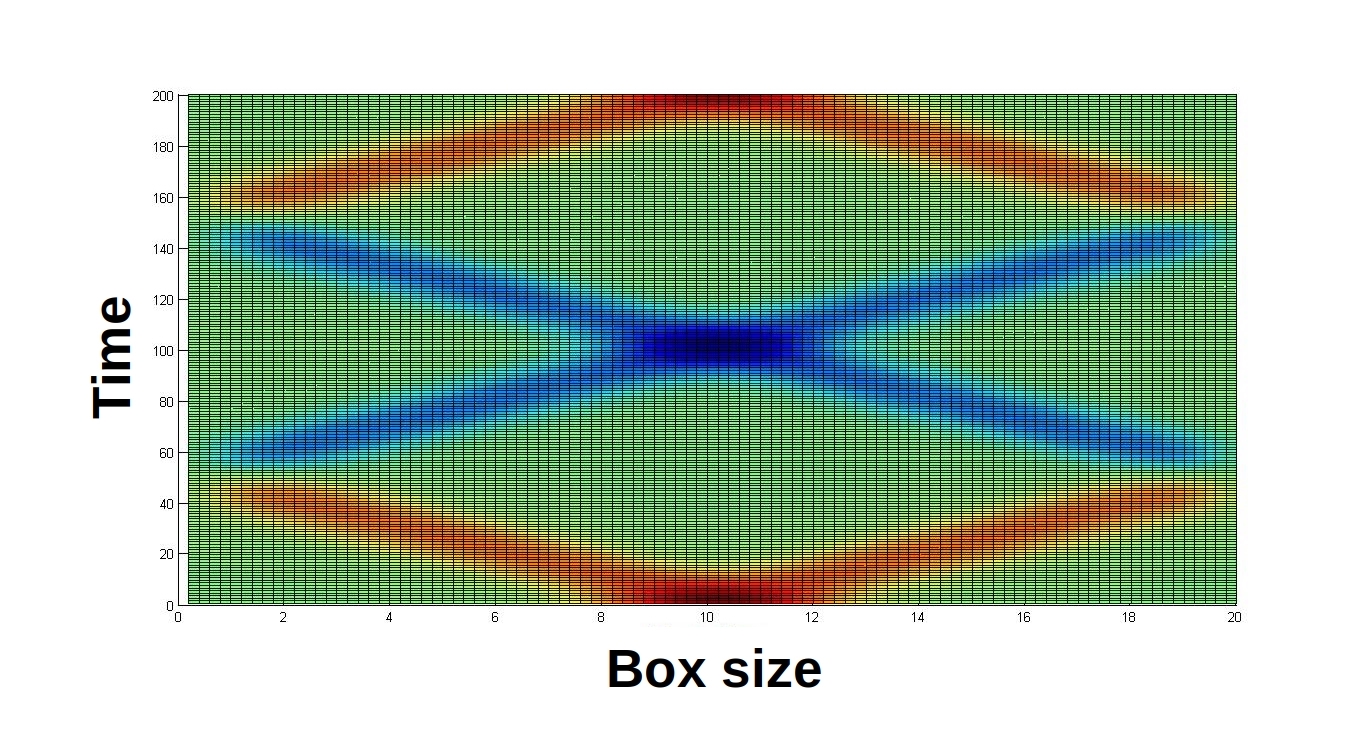}
		&\includegraphics[scale=0.18]{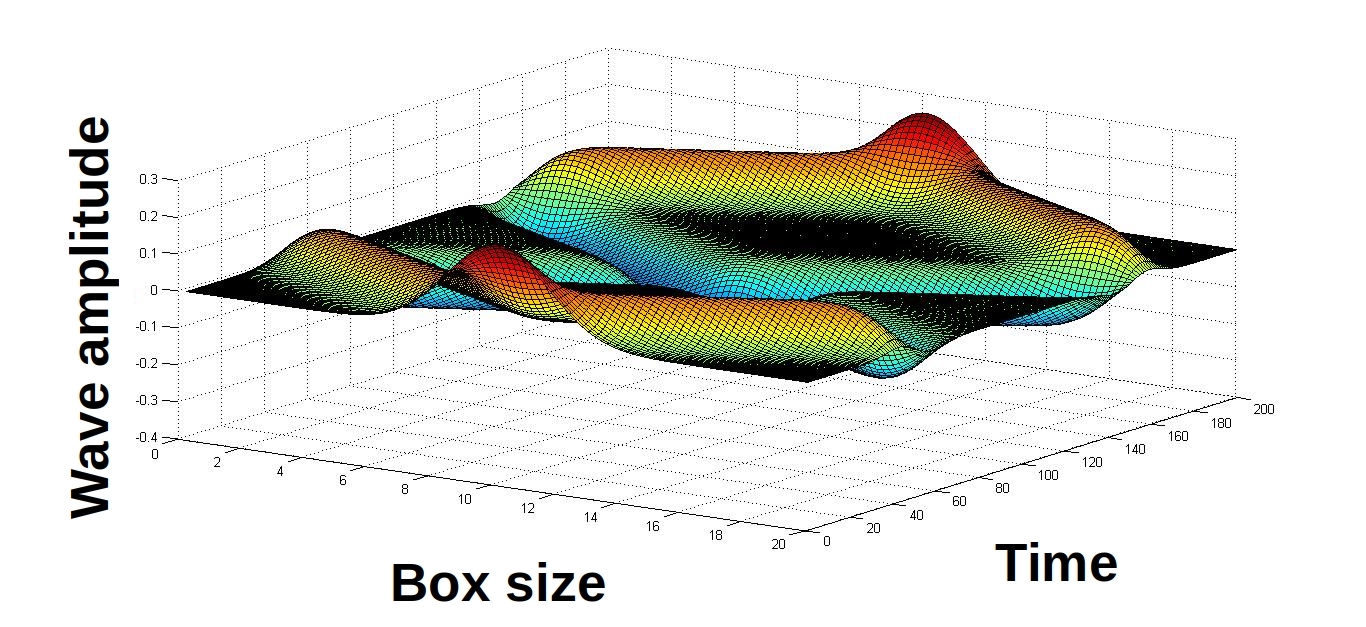}\\
		a)
		&
		b)
	\end{tabular}
	\caption{{\small \label{fig:spreading_wave} \textbf{Spreading wave on line segment Dirichlet.} In these figures we kept with the same parameters used for the previous plots, changing only the initial condition for $\vec{\phi}_{E}$. Now we can see the wave spreading equally for the both sides, reflecting in the boundary and then meeting themselves again in the center, but with the amplitude inverted. The units are the same used in the previous plots, meters and seconds.} }
\end{figure}

\begin{figure}[ht]
	\begin{tabular}{cc}
		\includegraphics[scale=0.15]{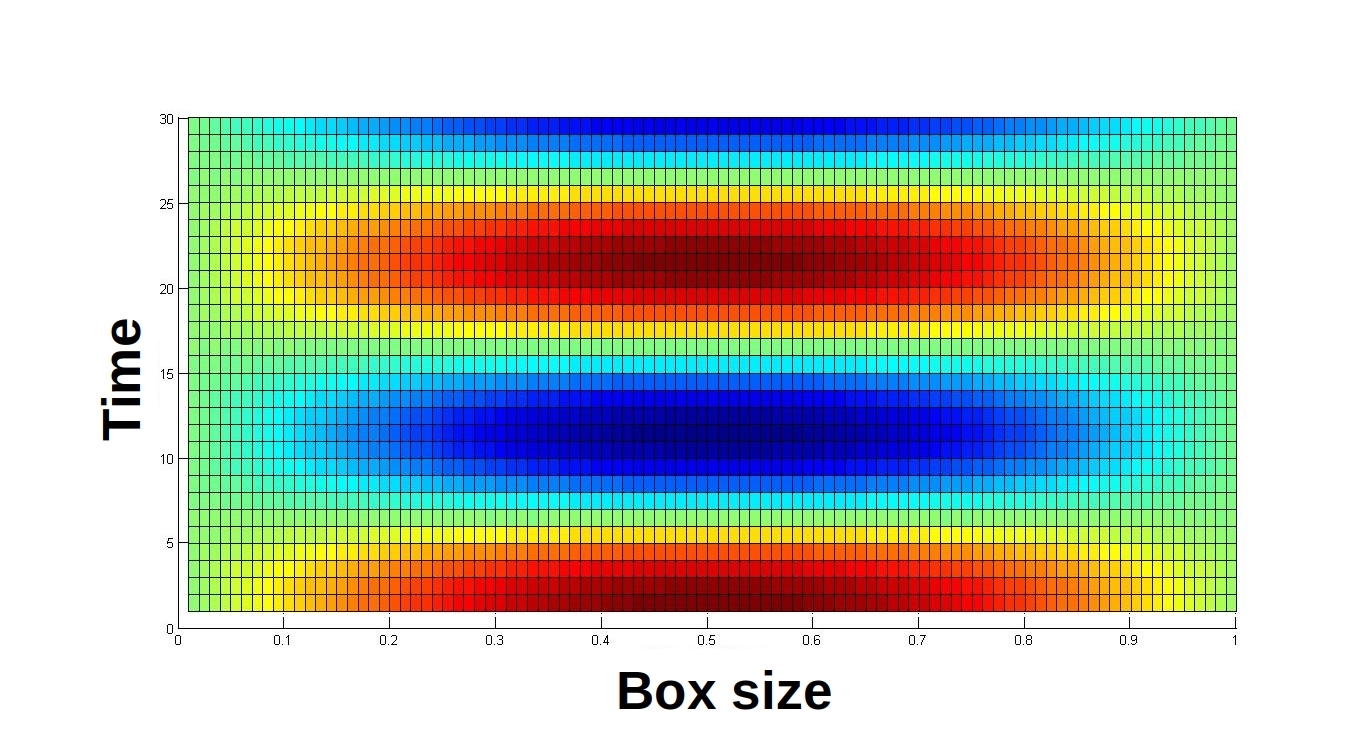}
		&\includegraphics[scale=0.18]{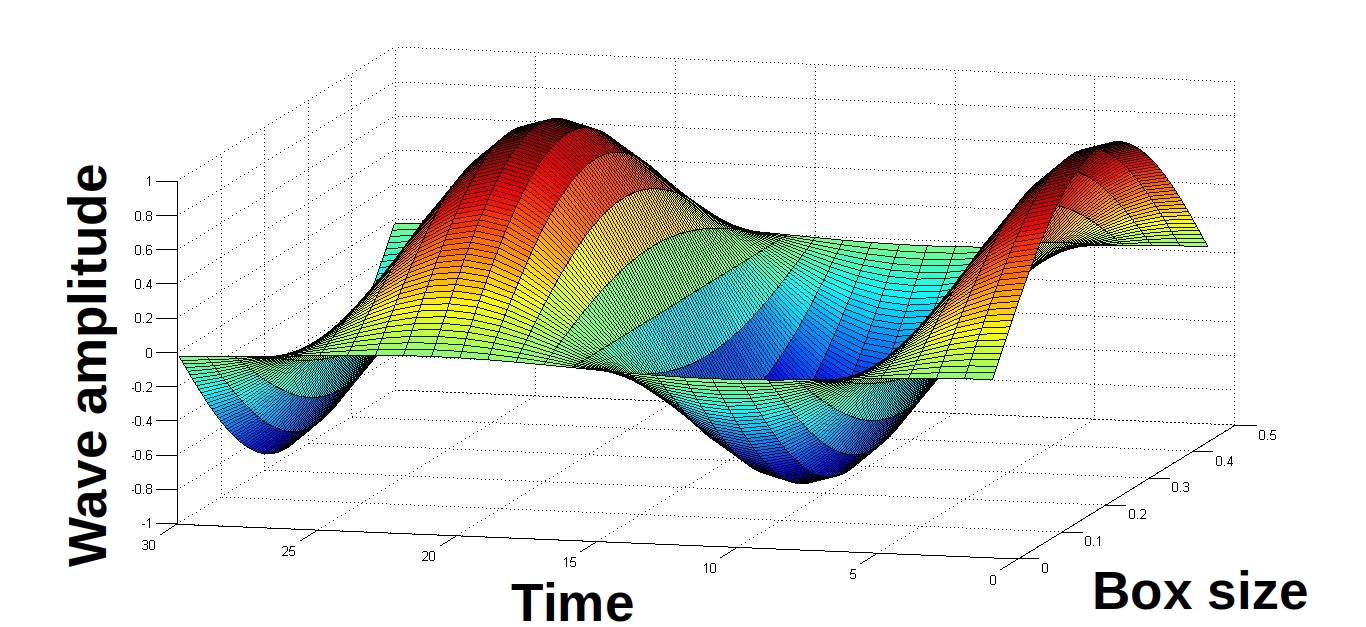}\\
		a)
		&
		b)
	\end{tabular}
	\caption{{\small \label{fig:standing_wave} \textbf{Standing wave.} Here we consider a standing wave, which can be described analytically by $\phi\left(x,t\right)=\cos\left(\omega t\right)\sin\left(\pi x\right)$. This can be simulated by Schr\"{o}dinger's equation if we work with $\vec{\phi}_{0}=\sin\left(\pi x\right)$ and $d \vec{\phi}_{0}/dt=0$ as long as we start with $t=0$. The units are the same ones used in the previous figures. 
	} }
\end{figure}

\begin{figure}[ht]
	\begin{center}
		\begin{tabular}{c}
			\includegraphics[scale=0.40]{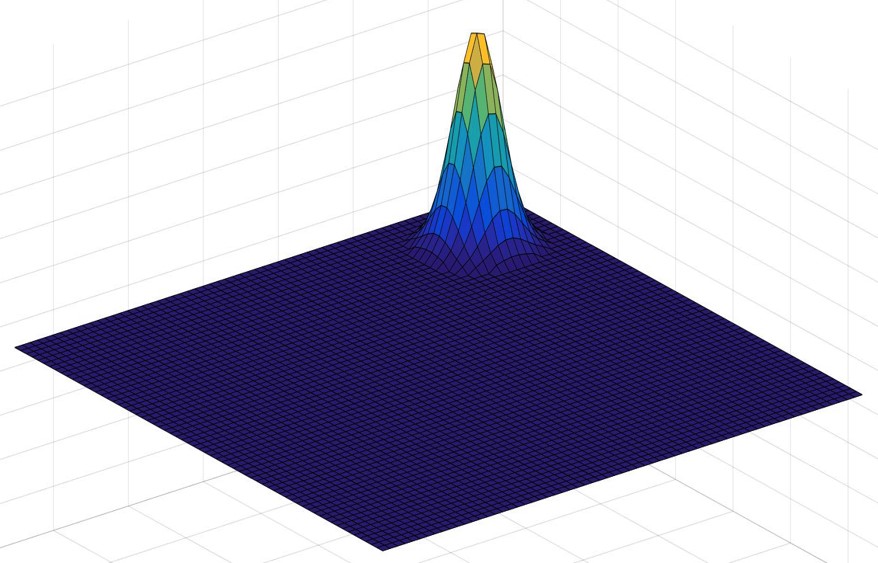}
			a)
		\end{tabular}
		\begin{tabular}{c}
			\includegraphics[scale=0.40]{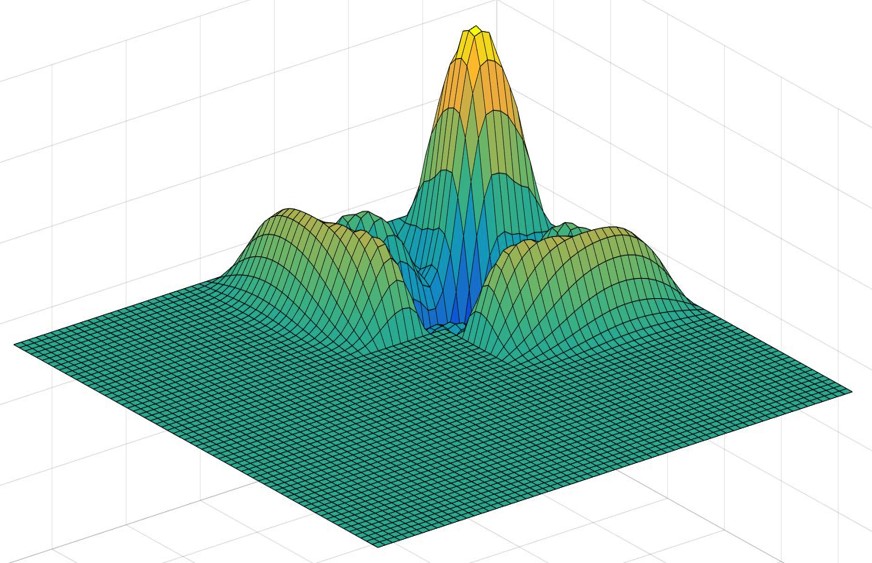}
			b)
		\end{tabular}
		\begin{tabular}{c}
			\includegraphics[scale=0.40]{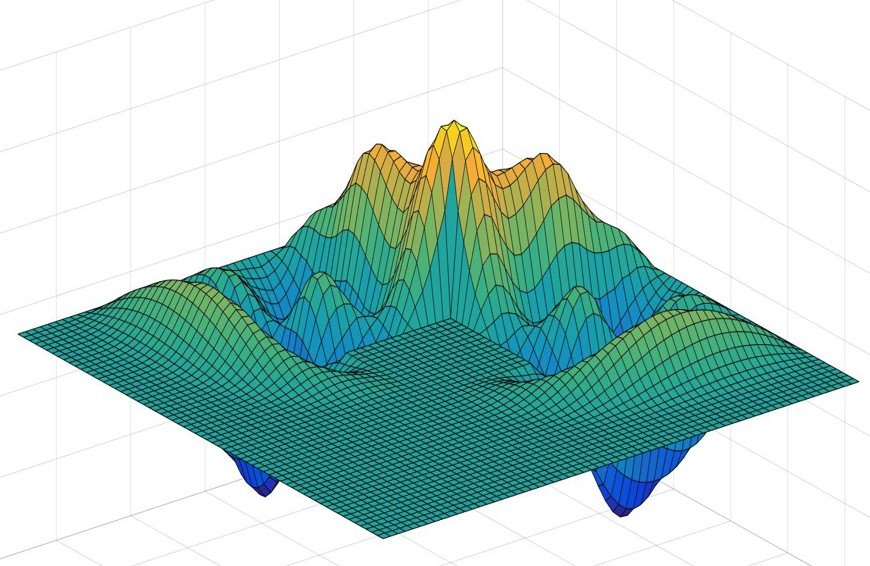}
			c)
		\end{tabular}
		\begin{tabular}{c}
			\includegraphics[scale=0.40]{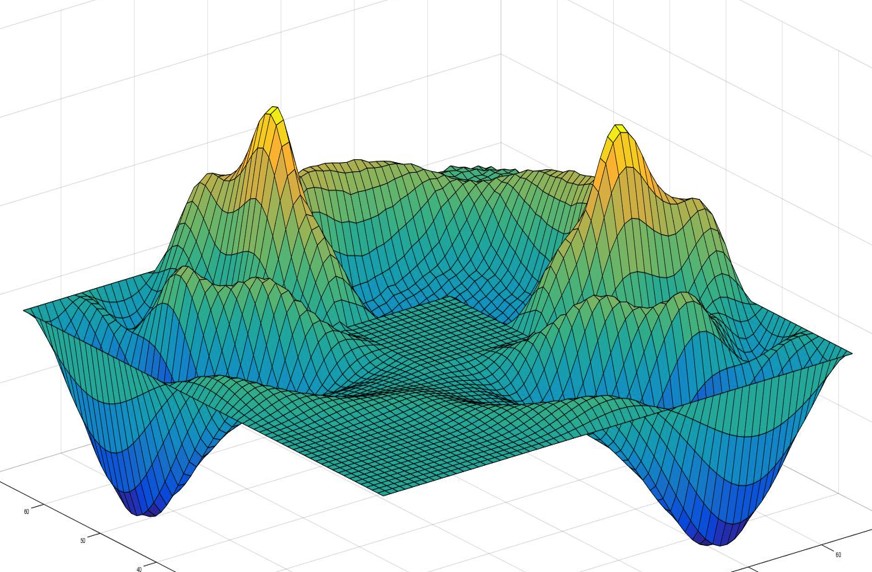}
			d)
		\end{tabular}
		\caption{{\small \label{fig:hole} \textbf{Wave packet in a cavity.} Here the initial state is a Gaussian wave packet, but now in a two dimensional region with nontrivial boundary. Specifically, we simulate scattering of the wave packet off a square object with Dirichlet boundary conditions. This is implemented as a square hole in the underlying discrete lattice.
				These four views represent the same wave packet in different time instants, where $t_{a}>t_{b}>t_{c}>t_{d}$. As in the one dimensional example, we worked with Dirichlet boundary conditions; however, the shape is not preserved. Here, the box has size ten in both axes, and we choose $a=0.1563$ and $\sigma=0.4$.}}
	\end{center}
\end{figure}


\section{Higher Order Laplacians}
\label{sec:highorder}

As we have seen, the graph Laplacian is related by a multiplicative constant to the first order approximation of the continuous Laplacian operator; however, higher order approximations might be desired to improve accuracy. In \cite{colbert1992novel} the authors give an expression for a finite difference approximation of the Laplacian operator that is based on the Lagrange interpolation formula and can be taken to arbitrarily high orders of accuracy.

The Lagrange interpolation formula is an exact formula for fitting a polynomial to a set of points $\{ x_i,f(x_i)=f_i \}$. For $2N+1$ values of $x_j$ labelled by $ j \in \{-N, -N+1, \dots N \}$, the formula is

\begin{equation} \label{lagrange}
f(x) = \sum_{k=-N}^{N} f(x_k) \prod_{l = -N, l \neq k}^{N} \left( \frac{x-x_l}{x_k-x_l} \right).
\end{equation}

Taking the second derivative of this formula gives an interpolation formula for an approximation of the Laplace operator. Assuming the values $x_j$ are taken from a uniform lattice (i.e. $x_j=ja$ for $j \in \mathbb{Z} $), we can approximate the Laplacian of $f$ at $x_0$ using
\begin{align} \label{lagrangelaplacian}
f''(x_0) & = \frac{-1}{a^2} \left( 2 f(x_0)\sum_{l=1}^{N} \frac{1}{l^2} - \sum_{k=1}^{N} \frac{f(x_k)+f(x_{-k})}{k^2} \prod_{l=-N, l \neq k}^{N}\frac{l^2}{l^2-k^2}  \right).
\end{align}

If we truncate this expression at $N=1$ then we recover the standard second order Laplacian approximation. (Recall that we define $k\th$ order to mean leading error term $O(a^k)$ on a lattice of spacing $a$.)

The next higher order ($N=2$) approximation of $f''(x_0)$ is
\begin{align}
f''(x_0) &= \frac{-1}{a^2} \left( \frac{5}{2}f(x_0) - \frac{4}{3} (f(x_1) + f(x_{-1})) + \frac{1}{12} (f(x_2) + f(x_{-2})) \right).
\end{align}

Assuming the lattice has periodic boundary conditions, then similar formulas hold at points other than $x_0$. In particular, we can write the fourth order Laplacian for a periodic lattice as
\begin{align}
L &= (-1/a^2) ((5/2) 1 -(4/3)(S+S^{\dagger})+(1/12)(S^2 +(S^{\dagger})^2) ).
\end{align}

Here $S$ is the matrix representation of the cyclic permutation $(1 2 \dots N)$, i.e., it has entries $S_{i,j} = \delta_{i,j+1 \mod N}$.

\clearpage

\section{Boundary Conditions for Higher Order Laplacians}
\label{sec:higherbc}

We can accomodate Neumann and Dirichlet boundary conditions by modifying Laplacians for periodic boundary conditions. In particular we follow the algebraic derivation described in Sec. \ref{sec:bc}. 

\subsection{Dirichlet Boundary Conditions}

As before, we consider a small neighborhood of vertices around $0$. By imposing that $\phi_j = 0$ for all $j \in \mathbb{Z}^-$, we modify the Laplacian as below.
\begin{eqnarray}
L \vec \phi &=& \begin{bmatrix}
5/2 & -4/3 & 1/12 & 0 & 0\\
-4/3 & 5/2 & -4/3 & 1/12 & 0\\
1/12 & -4/3 & 5/2 & -4/3 & 1/12\\
0 & 1/12 & -4/3 & 5/2 & -4/3 \\
0 & 0 & 1/12 & -4/3 & 5/2 
\end{bmatrix}
\begin{bmatrix}
\phi_{-2} \\
\phi_{-1} \\
\phi_{0} \\
\phi_{1} \\
\phi_{2}
\end{bmatrix}\nonumber\\
& \rightarrow & \begin{bmatrix}
0 & 0 & 0 & 0 & 0\\
0 & 0 & 0 & 0 & 0\\
0 & 0 & 5/2 & -4/3 & 1/12\\
0 & 0 & -4/3 & 5/2 & -4/3 \\
0 & 0 & 1/12 & -4/3 & 5/2 
\end{bmatrix}\nonumber
\begin{bmatrix}
0 \\
0 \\
\phi_{0} \\
\phi_{1} \\
\phi_{2}
\end{bmatrix}\nonumber\\
\end{eqnarray}

So imposing Dirichlet boundary conditions simply amounts to taking a principal submatrix.

\subsection{Neumann Boundary Conditions}

To account for Neumann boundary conditions, impose $\phi_j = \phi_0$ for all $j \in \mathbb{Z}^-$. The Laplacian is modified as below.
\begin{eqnarray}
L \vec \phi &=& \begin{bmatrix}
5/2 & -4/3 & 1/12 & 0 & 0\\
-4/3 & 5/2 & -4/3 & 1/12 & 0\\
1/12 & -4/3 & 5/2 & -4/3 & 1/12\\
0 & 1/12 & -4/3 & 5/2 & -4/3 \\
0 & 0 & 1/12 & -4/3 & 5/2 
\end{bmatrix}
\begin{bmatrix}
\phi_{-2} \\
\phi_{-1} \\
\phi_{0} \\
\phi_{1} \\
\phi_{2}
\end{bmatrix}\\
& \mapsto & \begin{bmatrix}
5/2 & -4/3 & 1/12 & 0 & 0\\
-4/3 & 5/2 & -4/3 & 1/12 & 0\\
1/12 & -4/3 & 5/2 & -4/3 & 1/12\\
0 & 1/12 & -4/3 & 5/2 & -4/3 \\
0 & 0 & 1/12 & -4/3 & 5/2 
\end{bmatrix}
\begin{bmatrix}
\phi_{0} \\
\phi_{0} \\
\phi_{0} \\
\phi_{1} \\
\phi_{2}
\end{bmatrix}\\
& = & \begin{bmatrix}
0 \\
(5/2-4/3-4/3+1/12) \phi_0 +(1/12) \phi_1 \\
(5/2-4/3+1/12) \phi_0 -(4/3) \phi_1 + (1/12) \phi_2 \\
\dots
\end{bmatrix}\\
& = & \begin{bmatrix}
0 & 0 & 0 & 0 & 0\\
0 & 0 & -1/12 & 1/12 & 0\\
0 & 0 & 5/4 & -4/3 & 1/12\\
0 & 0 & -5/4 & 5/2 & -4/3 \\
0 & 0 & 1/12 & -4/3 & 5/2 \\
\end{bmatrix}
\begin{bmatrix}
\phi_{0} \\
\phi_{0} \\
\phi_{0} \\
\phi_{1} \\
\phi_{2}
\end{bmatrix}
\end{eqnarray}

Note that this is not a symmetric approximation of the Laplacian, not even when restricted to vertices $0,1,$ and $2$. However the decoupled second order dynamics of Eqn. \ref{decouple} require symmetric operators since $B B^{\dagger}$ is Hermitian by construction, so our algorithm cannot use higher order Laplacians for simulating dynamics with Neumann boundary conditions.


\subsection{Hypergraph Incidence Matrices}

Now that we've seen how to impose Dirichlet boundary conditions on higher order Laplacians, we should consider how to generate their incidence matrices. Recall that the fourth order Laplacian with periodic boundary conditions is
\[
L = (-1/a^2) ((5/2) \id -(4/3)(S+S^{\dagger})
+(1/12)(S^2 +(S^{\dagger})^2) ),
\] 
which is a sum of circulant matrices. This suggests that a reasonable ansatz for the incidence matrix is $cS-(c+b) \id +bS^{\dagger} $. 
By construction this ansatz has zero sum rows which guarantees that the Laplacian matrix acting on a vector whose entries all have the same value will evaluate to 0 (which is consistent with the fact that the Laplacian operator acting on a constant function evaluates to 0). 

From this ansatz we arrive at the following system of degree 2 polynomial equations in $b$ and $c$.

\begin{align} \label{asys}
2(c^2+b^2+cb) & =  5/2 \\
cb & =  1/12 \\
(c+b)^2 & =  4/3
\end{align}

Once any two of these is satified the third will also be satisfied since the row sums of the matrix must all be zero. The middle equation gives us $b=1/12c$, which substituted into the last equation gives $4/3 = c^4  -(7/6)c^2 + (1/144)$ which has solutions satisfying $c^2  =  (7/12) \pm \sqrt{1/3}$. This gives values of $c \approx 1.07735$ and $b \approx 0.07735$ (switching their values gives another solution). 

\subsection{2 Dimensions and Beyond} \label{2dandbeyond}

The continuous Laplacian in 2 dimensions can be written as $\nabla^2 = \frac{\partial^2}{\partial x^2} +  \frac{\partial^2}{\partial y^2}$, i.e. the sum of the one dimensional Laplacians in the $x$ and $y$ directions (note that each of these is basis dependent although the total Laplacian is not). Discrete Laplacians in 2 dimensions are similarly constructed.

We discretize space into a square lattice and remove some edges and vertices according to boundary conditions. The resulting graph $(V,E)$ is a subgraph of the square lattice, so we can separate its edge set into vertical edges, $E_y$, and horizontal edges, $E_x$. The subgraphs associated with this partition, $G_x=(V,E_x)$ and $G_y=(V,E_y)$, are composed of several disconnected path graphs (or cycles under periodic boundary conditions). 
If the lattice is $n$ vertices wide and $m$ vertices tall then $G_x$ consists of $m$ path graphs each on $n$ vertices; similarly $G_y$ consists of $n$ path graphs each on $m$ vertices. If scatterers are introduced then the path graphs composing $G_x$ and $G_y$ will depend on what edges and vertices are removed to account for the scatterers.

Since $G_x$ and $G_y$ are composed of several disconnected path graphs, we can write down their Laplacians and factor them into incidence matrices. The Laplacians $L(G_x)$ and $L(G_y)$ approximate $\frac{\partial^2}{\partial x^2}$ and $\frac{\partial^2}{\partial y^2}$, respectively, so $L(G_x)+L(G_y)$ approximates $\nabla^2$. If $L(G_x)=B_x^{\dagger} B_x$ and $L(G_y)=B_y^{\dagger} B_y$, then $L(G_x)+L(G_y) = C^{\dagger} C$ where $C$ is the $|E_x \cup E_y| \times V$ matrix produced by vertically concatenating $B_x$ and $B_y$.

Generalizing this to $n$-dimensions, the procedure is (1) separate the lattice into $n$ graphs (each composed of disconnected paths or cycles) corresponding to each direction in space (2) write down the Laplacians for these $n$ graphs and factor them into incidence matrices and (3) vertically concatenate their incidence matrices.

\subsection{Sixth (and higher) Order Laplacians}
So far our discussion has been restricted to second and fourth order Laplacians; however, we can arrive at higher order Laplacians by (1) taking higher order expansions of the Lagrange interpolation formula, (2) differentiating twice and evaluating at $x=0$, and (3) reading off the interpolation formula coefficients as matrix coefficients. Periodic boundary conditions are achieved by requiring that the Laplacian be circulant.
As before, Dirichlet boundary conditions can be imposed by taking principal submatrices of the Laplacian. Our remarks about generalizing beyond 1-D also hold for higher order Laplacians.

The problem of finding the incidence matrices of higher order Laplacians is a little more involved that in the 1st order case where the graph theoretic interpretation facilitates the factorization. We let $N$ denote the radius of a discrete Laplacian. That is, a Laplacian matrix with nonzero entries only out to nearest neighbors has $N=1$, second nearest neighbors has $N=2$, and so on. In general the radius $N$ Laplacian will be factored into incidence matrices of hypergraphs where each hyperedge can contain up to $N+1$ vertices. (Note: Hyperedges with fewer than $N+1$ vertices will appear if Dirichlet boundary conditions are used.)

As in the $N=1$ and $N=2$ (\emph{i.e.} second and fourth order) cases, the entries of these incidence matrices can be found by considering the factorization of a Laplacian with periodic boundary conditions. The translational invariance of this case guarantees that all hyperedges will have the same weights and can be oriented identically. Then the entries of the incidence matrix can be found by choosing an appropriate ansatz (one of the form $\sum_{j=-m}^n a_j S^j$ for some $n$ and $m$) and solving the appropriate system of polynomial equations (similar to how \ref{asys} was solved). We provide numerical values for the entries of Laplacians and their incidence matrices up to tenth order in appendix \ref{ordertables}.


\section{Discretization Errors}
\label{sec:errors}

Using a $k^{\mathrm{th}}$ order Laplacian, as described in \sect{sec:highorder} one expects discretization errors to shrink with lattice spacing as $O(a^k)$. To obtain a more quantitative assessment of discretization errors, we can numerically compute a metric called the \emph{Q factor}, which is used to quantify discretization errors in numerical simulations \cite{mexico}.

To compute this factor we use the discretized solutions at three different lattice spacings $\varPhi^{a}$ , $\varPhi^{2a}$ and $\varPhi^{4a}$. The Q factor is then defined by
\begin{equation}
\label{Qfactor}
Q\left(t\right)=\frac{\left\Vert \varPhi^{4a}-\varPhi^{2a}\right\Vert _{2}}{\left\Vert \varPhi^{2a}-\varPhi^{a}\right\Vert _{2}}.
\end{equation}
$\varPhi^{4a}$ and $\varPhi^{2a}$ are defined on different lattices, and thus they are vectors of different dimension. However, we choose the lattices so that the vertices present in the lattice of spacing $4a$ are a subset of the vertices present in the lattice of spacing $2a$. Then, by $\left\Vert \varPhi^{4a}-\varPhi^{2a}\right\Vert$ we really mean the $l_2$ norm of the vector $\varPhi^{4a}-I_{4a}(\varPhi^{2a})$, where $I_{4a}$ is the inclusion map that discards the vector components associated with the vertices absent from the lattice of spacing $4a$. For notational simplicity we drop explicit reference to this inclusion map.

Now we want to see the value associated with Eq.(\ref{Qfactor}) when we take the continuum limit, $a\rightarrow0$. Straightforward Taylor expansion shows that a $k^{\textrm{th}}$ order discretized Laplacian, which leaves errors of order $a^k$ should yield a corresponding Q factor of $2^k$ in the limit of $a$ going to zero, provided errors from other steps in the algorithm, such as state preparation do not dominate. Now we present a table of values that shows the average of $Q$ from $t=0$ to $t=0.5$, working with $0.0001$ as the time step.
\noindent \begin{center}
	\begin{tabular}{|l|c|c|}
		\hline
		& Second order & Fourth order\tabularnewline
		\hline
		\hline
		$\left\langle Q\right\rangle_{\text{spreading}}$ &3.98 & 15.69 \tabularnewline
		\hline
		$\left\langle Q\right\rangle_{\text{rigidly-translating}}$ &1.99 &2.00\tabularnewline
		\hline
		$\left\langle Q\right\rangle_{\text{standing}}$ & 3.99 & 15.89\tabularnewline
		\hline
	\end{tabular}
	\par\end{center}

One sees that for the spreading wave packet case and the standing wave case (both static initial conditions) the Q factors are in good agreement with the expected values of 4 and 16 for the second-order and fourth-order Laplacians. In the case of the rigidly translating wave packet (which corresponds to the initial condition of \eq{rigid1d}), the Q factor is approximatly 2 independent of the order of the discretized Laplacian. This is because, in this case, the dominant source of error is in the state preparation. Exact state preparation would involve inverting the incidence matrix, as described in \sect{sec:general}. The initial state described by \eq{rigid1d} is accurate only up to errors of order $a$, thus yielding a Q factor of 2. In appendix \ref{Q_statics} we also obtain an analytical calculation of the Q factor for the special case of a standing wave, treated with a first-order Laplacian.

Since a $k^{\textrm{th}}$ order Laplacian gives truncation errors of order $a^k$, the total error accumulated for evolution time $T$ will be order $a^kT$. A $D$ dimensional Laplacian of order $k$ has an incidence matrix which is $D(k/2+1)$-sparse; so if an $s$-sparse Hamiltonian is used then $k=2(s/D)-2$. Then the total error accumulated is on the order of $Ta^{2(s/D)-2}$.

\section{Smoothness}
\label{sec:smoothness}

In preceding sections we have discussed the impact of using higher order discretizations to minimize error. In general, both classically and quantumly, one chooses the order of the discretization of the Laplacian on a lattice to obtain discretization errors of order $a^k$, where $a$ is the lattice spacing. The choice of $k$ is influenced by the smoothness of the underlying continuum solution that one is attempting to discretize. A high order discretization with error $O(a^k)$ of an $m\th$ derivative is only justified if the exact solution is $(k+m)$-times differentiable, since any such discretization of an $m\th$ derivative is derived by Taylor expanding the exact solution to order $k+m$. Furthermore, knowing the magnitude of these higher derivatives allows quantitative error bounds to be derived, as we show in this section.

\begin{theorem}
	\label{blowup_lemma}
	Let $\Omega$ be a bounded convex domain in $\mathbb{R}^d$. Let $f$ be a smooth function on $\Omega$ that vanishes on the boundaries. Let $\vec{v}(\vec{x})$ be the solution to
	\begin{equation}
	\label{gausseq}
	\vec{\nabla} \cdot \vec{v}(\vec{x}) = f(\vec{x})
	\end{equation}
	on $\Omega$ with no divergenceless component. Then,
	\begin{equation}
	\label{blowup_bound}
	\sqrt{\int_\Omega d^d x \ \vec{v}(\vec{x}) \cdot \vec{v}(\vec{x})} \leq \frac{\ell}{\pi} \sqrt{\int_\Omega d^d x \ f(\vec{x})^2}
	\end{equation}
	where $\ell$ is the diameter of $\Omega$.
\end{theorem}

\begin{proof}
	The divergence operator is not invertible because it has a kernel. However, it does have a Moore-Penrose pseudoinverse $\textrm{Div}^{-1}$, which is typically expressed in terms of the Green's function, as follows.
	\begin{equation}
	\label{divGreens}
	\textrm{Div}^{-1}[f](\vec{x}) = \int_{\Omega} d^d y \ f(\vec{y}) \frac{\vec{y} - \vec{x}}{|\vec{y}-\vec{x}|^d}.
	\end{equation}
	We next note that the Laplacian operator can be written as $\nabla^2 = \nabla^\dag \nabla$. (Here, we think of $\nabla$ as a column vector of partial derivative operators.) The singular values of the Laplacian are therefore the squares of the singular values of $\nabla^\dag$, which is the Divergence operator. The fundamental gap theorem \cite{PW60, AC11} states that on a convex bounded domain $\Omega$, the smallest nontrivial eigenvalue of the Laplacian subject to Neumann boundary conditions is lower bounded by $\pi^2/\ell^2$ where $\ell$ is the diameter of $\Omega$. Consequently, the smallest nonzero singular value of $\nabla^\dag$, \emph{i.e.} the divergence operator, can be at most $\pi/\ell$. Hence the largest singular value of $\textrm{Div}^{-1}$ can be at most $\ell/\pi$. Thus we obtain \eq{blowup_bound}.
\end{proof}

\begin{theorem}
	\label{accumulation}
	Let $\mathcal{D}$ be Hermitian linear combination of finite-order partial derivatives on $\mathbb{R}^d$. Let $\phi_{\lambda}$ be the solution to
	\begin{equation}
	\label{lambdawave}
	\frac{\partial^2 \phi_\lambda}{\partial t^2} = \nabla^2 \phi_{\lambda} - \lambda^2 \mathcal{D}^2 \phi_\lambda
	\end{equation}
	on some continuous domain $\Omega \subset \mathbb{R}^d$ subject to Dirichlet or Neumann boundary conditions. We take initial conditions at $t=0$ to be fixed functions $\phi(\vec{x},0)$, and $\dot{\phi}(\vec{x},0)$ independent of $\lambda$. Then for any $\epsilon \in \mathbb{R}$ and any $t \geq 0$
	\begin{equation}
	\label{generalbound}
	\|\phi_{\epsilon}(t)-\phi_{0}(t)\|\leq\sqrt{2t\epsilon}\left[\left(\|\phi(0)\|^{2}+\sum_{j=1}^{d}\|\|\psi_{j}(0)\|^{2}\right) \left(\|\mathcal{D}\phi(0)\|^{2}+\sum_{j=1}^{d}\|\mathcal{D}\psi_{j}(0)\|^{2}\right)\right]^{1/4}
	\end{equation}
	where $\| f \| \equiv \sqrt{\int_\Omega d^d x |f(\vec{x})|^2}$ and
	\begin{equation}
	\vec{\psi}(\vec{x}, 0) = \int d^d y \frac{\vec{x}-\vec{y}}{|\vec{x} - \vec{y}|^d} \dot{\phi}(\vec{y},0).
	\end{equation}
\end{theorem}

\begin{proof}
	Let
	
	\begin{eqnarray}
	S_\lambda & = & \left[ \begin{array}{c} \phi_\lambda \\ \vec{\psi}_\lambda \end{array} \right] \label{slambda}\\
	H_0 & = & \left[ \begin{array}{cccc}
	0 & \frac{\partial}{\partial x_1} & \hdots & \frac{\partial}{\partial x_d} \\
	-\frac{\partial}{\partial x_1} & 0 & \hdots & 0 \\
	\vdots & 0 & \hdots & 0 \\
	-\frac{\partial}{\partial x_d} & 0 & \hdots & 0 \end{array} \right] \\
	H_{\mathcal{D}} & = & \left[ \begin{array}{cccc}
	\mathcal{D} & 0 & \hdots & 0 \\
	0 & -\mathcal{D} & \hdots & 0 \\
	\vdots & & \ddots & \\
	0 & 0 & \hdots & -\mathcal{D} \end{array} \right] \\
	H_\lambda &  = & H_0 + \lambda H_{\mathcal{D}} \\
	\frac{dS_{\lambda}}{dt} & = & -i H_\lambda S_\lambda. \label{fake_Schrod}
	\end{eqnarray}
	By \eq{fake_Schrod},
	\begin{eqnarray}
	\frac{d^2}{dt^2} S_\lambda & = & - H_\lambda^2 S_\lambda \\
	& = & \left[ \begin{array}{cccc}
	\nabla^2 - \lambda^2 \mathcal{D}^2 & 0 & \ldots & 0 \\
	0 & \frac{\partial^2}{\partial x_1^2} - \lambda^2 \mathcal{D}^2 & \ldots & \frac{\partial^2}{\partial x_1 \partial x_d} \\
	\vdots & & \ddots & \\
	0 & \frac{\partial^2}{\partial x_d \partial x_1} & \ldots & \frac{\partial^2}{\partial x_d^2} - \lambda^2 \mathcal{D}^2
	\end{array} \right] \left[ \begin{array}{c} \phi \\ \psi_1 \\ \vdots \\ \psi_d \end{array} \right] \nonumber. 
	\end{eqnarray}
	Thus the solution to \eq{fake_Schrod} satisfies \eq{lambdawave}. As initial conditions $(t=0)$ for $\vec{\psi}_\lambda$ we can take
	\begin{equation}
	\vec{\psi}(0) = \textrm{Div}^{-1} \left[ \dot{\phi}(0) \right]
	\end{equation}
	where $\textrm{Div}^{-1}$ is as defined in \eq{divGreens}. 
	By \eq{fake_Schrod} we have
	\begin{eqnarray}
	\frac{d}{dt}\langle S_{\epsilon},S_{0}\rangle&=&\left\langle \dot{S}_{\epsilon}(t),S_{0}(t)\right\rangle +\left\langle S_{\epsilon}(t),\dot{S}_{0}(t)\right\rangle \\
	&=&\langle-i(H_{0}+\epsilon H_{\mathcal{D}})S_{\epsilon}(t),S_{0}(t)\rangle+\langle S_{\epsilon}(t),-iH_{0}S_{0}(t)\rangle 
	\\
	&=&\langle S_{\epsilon}(t),i(H_{0}+\epsilon H_{\mathcal{D}})S_{0}(t)\rangle+\langle S_{\epsilon}(t),-iH_{0}S_{0}(t)\rangle \\
	&=&i\epsilon\langle S_{\epsilon}(t),H_{\mathcal{D}}S_{0}(t)\rangle.
	\end{eqnarray}
	Thus, by the Cauchy-Schwarz inequality
	\begin{equation}
	\left| \frac{d}{dt} \langle S_\epsilon, S_0 \rangle \right| \leq \epsilon \| S_\epsilon(t) \| \times \| H_{\mathcal{D}} S_0(t) \|,
	\end{equation}
	where $\| S \|$ is a shorthand for $\sqrt{ \langle S, S \rangle }$. $H_\lambda$ is Hermitian for real $\lambda$ and therefore $\| S_\epsilon(t) \| = \| S_\epsilon(0) \|$.
	\begin{equation}
	\label{chbound1}
	\left| \frac{d}{dt} \langle S_\epsilon, S_0 \rangle \right| \leq \epsilon \| S_\epsilon(0) \|  \times \| H_{\mathcal{D}} S_0(t) \|.
	\end{equation}
	Next, we observe that
	\begin{equation}
	\| H_{\mathcal{D}} S_0(t) \| = \| H_{\mathcal{D}}^{+} S_0(t) \|
	\end{equation}
	where the operator
	\begin{equation}
	H_{\mathcal{D}}^{+} = \left[ \begin{array}{cccc}
	\mathcal{D} & 0 & \ldots & 0 \\
	0 & \mathcal{D} & \ldots & 0 \\
	\vdots & & \ddots & \\
	0 & 0 & \ldots & \mathcal{D} \end{array} \right]
	\end{equation}
	is Hermitian and commutes with $H_0$. Thus, $\| H_{\mathcal{D}} S_0(t) \|$ is conserved, and \eq{chbound1} becomes
	\begin{equation}
	\label{chbound2}
	\left| \frac{d}{dt} \langle S_\epsilon, S_0 \rangle \right| \leq \epsilon \| S_\epsilon(0) \|  \times \| H_{\mathcal{D}} S_0(0) \|,
	\end{equation}
	which expands out to
	\begin{equation}
	\label{changebound}
	\left| \frac{d}{dt} \langle S_\epsilon, S_0 \rangle \right| \leq \epsilon \sqrt{ \| \phi(0) \|^2 + \sum_{j=1}^d \| \psi_j(0) \|^2 }
	\sqrt{ \| \mathcal{D} \phi(0) \|^2 + \sum_{j=1}^d \| \mathcal{D} \psi_j(0) \|^2 }.
	\end{equation}
	By definition
	\begin{eqnarray}
	\|S_{\epsilon}(t)-S_{0}(t)\|&=&\langle S_{\epsilon}(t)-S_{0}(t),S_{\epsilon}(t)-S_{0}(t)\rangle\\&=&\langle S_{\epsilon}(t),S_{\epsilon}(t)\rangle+\langle S_{0}(t),S_{0}(t)\rangle -2\mathrm{Re}\langle S_{\epsilon},S_{0}\rangle
	\end{eqnarray}
	The ``Hamiltonians'' $H_0$ and $H_{\mathcal{D}}$ are Hermitian so $\langle S_\epsilon(t), S_\epsilon(t) \rangle$ and  $\langle S_0(t), S_0(t) \rangle$ are time-independent. Thus,
	\begin{equation}
	\frac{d}{dt} \| S_\epsilon(t) - S_0(t)\|^2 =  -2 \mathrm{Re} \frac{d}{dt} \langle S_\epsilon, S_0 \rangle.
	\end{equation}
	Applying \eq{changebound} yields
	\begin{equation}
	\label{changebound2}
	\left|\frac{d}{dt}\|S_{\epsilon}(t)-S_{0}(t)\|^{2}\right|\leq2\epsilon\sqrt{\|\phi(0)\|^{2}+\sum_{j=1}^{d}\|\psi_{j}(0)\|^{2}}\sqrt{\|\mathcal{D}\phi(0)\|^{2}+\sum_{j=1}^{d}\|\mathcal{D}\psi_{j}(0)\|^{2}}.
	\end{equation}
	The triangle inequality and \eq{changebound2} yield
	
	\begin{eqnarray}
	\|S_{\epsilon}(t)-S_{0}(t)\|^{2}&=&\|S_{\epsilon}(0)-S_{0}(0)\|^{2}+\int_{0}^{t}d\tau\frac{d}{d\tau}\|S_{\epsilon}(\tau)-S_{0}(\tau)\|^{2}\nonumber\\&\leq&\|S_{\epsilon}(0)-S_{0}(0)\|^{2}+\int_{0}^{t}d\tau\left|\frac{d}{d\tau}\|S_{\epsilon}(\tau)-S_{0}(\tau)\|^{2}\right|\nonumber\\&\leq&\|S_{\epsilon}(0)-S_{0}(0)\|^{2}+2t\epsilon\sqrt{\|\phi(0)\|^{2}+\sum_{j=1}^{d}\|\psi_{j}(0)\|^{2}}\sqrt{\|\mathcal{D}\phi(0)\|^{2}+\sum_{j=1}^{d}\|\mathcal{D}\psi_{j}(0)\|^{2}}.\nonumber
	\end{eqnarray}
	The initial conditions have $S_\epsilon(0) = S_0(0)$, and therefore
	\begin{equation}
	\label{Sdiff}
	|S_{\epsilon}(t)-S_{0}(t)\|^{2}\leq2t\epsilon\sqrt{\|\phi(0)\|^{2}+\sum_{j=1}^{d}\|\psi_{j}(0)\|^{2}}\sqrt{\|\mathcal{D}\phi(0)\|^{2}+\sum_{j=1}^{d}\|\mathcal{D}\psi_{j}(0)\|^{2}}.
	\end{equation}
	Recalling the definition of $S_\lambda$ \eq{slambda},
	\begin{equation}
	\| S_\epsilon(t) - S_0(t) \|^2 = \| \phi_\epsilon(t) - \phi_0(t) \|^2 + \| \vec{\psi}_\epsilon(t) - \vec{\psi}_0(t) \|^2. 
	\end{equation}
	Thus \eq{Sdiff} implies the bound
	\begin{equation}
	\|\phi_{\epsilon}(t)-\phi_{0}(t)\|^{2}\leq2t\epsilon\sqrt{\|\phi(0)\|^{2}+\sum_{j=1}^{d}\|\psi_{j}(0)\|^{2}}\sqrt{\|\mathcal{D}\phi(0)\|^{2}+\sum_{j=1}^{d}\|\mathcal{D}\psi_{j}(0)\|^{2}}.\nonumber
	\end{equation}
	From this we obtain the final bound.
\end{proof}
In the special case that $\mathcal{D} = \nabla^2$ and $\Omega$ is convex we can bound $\| H_\mathcal{D} S_0(0) \|$ in terms of more accessible quantities, as follows.
\begin{theorem}
	\label{accumulation_lap}
	Let $\phi_{\lambda}$ be the solution to
	\begin{equation}
	\frac{\partial^2 \phi_\lambda}{\partial t^2} = \nabla^2 \phi_{\lambda} - \lambda^2 \left( \nabla^2 \right)^2 \phi_\lambda
	\end{equation}
	on some convex domain $\Omega \subset \mathbb{R}^d$ subject to Dirichlet or Neumann boundary conditions. We take initial conditions at $t=0$ to be fixed functions $\phi(\vec{x},0)$, and $\dot{\phi}(\vec{x},0)$ independent of $\lambda$. Then for any $\epsilon \in \mathbb{R}$ and any $t \geq 0$
	\begin{equation}
	\label{lapbound}
	\|\phi_{\epsilon}(t)-\phi_{0}(t)\|\leq\sqrt{2t\epsilon\|\nabla^{2}\phi(0)\|}\left(\|\phi(0)\|^{2}+\frac{\ell^{2}}{\pi^{2}}\|\dot{\phi}(0)\|^{2}\right)^{1/4}.
	\end{equation}
	where $\| f \| \equiv \sqrt{\int_\Omega d^d x |f(\vec{x})|^2}$.
\end{theorem}

\begin{proof}
	By theorem \ref{accumulation},
	\begin{equation}
	\label{fromaccum}
	\|\phi_{\epsilon}(t)-\phi_{0}(t)\|\leq\sqrt{2t\epsilon}\left[\left(\|\phi(0)\|^{2}+\sum_{j=0}^{d}\|\psi_{j}(0)\|^{2}\right)\left(\|\nabla^{2}\phi(0)\|^{2}+\sum_{j=0}^{d}\|\nabla^{2}\psi_{j}(0)\|^{2}\right)\right]^{1/4}
	\end{equation}
	By theorem \ref{blowup_bound},
	\begin{equation}
	\label{firstnote}
	\sum_{j=0}^d \| \psi_j(0) \|^2 \leq \frac{\ell^2}{\pi^2} \| \dot{\phi}(0) \|.
	\end{equation}
	Recalling \eq{divGreens}, we have
	\begin{eqnarray}
	\nabla^2 \vec{\psi}(0) & = & 
	\nabla^2 \int d^d y \frac{\vec{x} - \vec{y}}{|\vec{x} - \vec{y}|^d} \dot{\phi}(\vec{y}) \\
	& = &  \int d^d y \left( \nabla^2 \frac{\vec{x} - \vec{y}}{|\vec{x} - \vec{y}|^d} \right) \dot{\phi}(\vec{y}) \\
	& = & \vec{0}. \label{zerovec}
	\end{eqnarray}
	Substituting \eq{firstnote} and \eq{zerovec} into \eq{fromaccum} yields \eq{lapbound}.
\end{proof}

Theorem \ref{accumulation} gives a very nice quantitative upper bound on discretization errors in terms of directly accessible properties of the initial conditions. However, it only applies under the specific condition that the error term of interest is expressible as a negative coefficient times the square of a Hermitian linear combination of partial derivatives. Not all discretized Laplacians satisfy this. However, it is possible to engineer high order Laplacians such that this is the case. We illustrate this by giving an explicit discretized Laplacian in two dimensions with error of order $a^2$, which satisfies this condition. The formula is
\begin{equation}
\begin{array}{l} \frac{1}{a^2} \left\{ -\frac{2}{15} \left[ \phi(x,y+2a) + \phi(x,y-2a) + \phi(x+2a,y) + \phi(x-2a,y) \right] \right. \vspace{5pt} \\
-\frac{1}{10}\left[\phi(x+a,y+a)+\phi(x-a,y+a)+\phi(x+a,y-a)+\phi(x-a,y-a)\right] \vspace{5pt} \\
+\frac{26}{15}\left[\phi(x+a,y)+\phi(x-a,y)\vspace{5pt}+\phi(x,y+a)\left.+\phi(x,y-a)\right]-6\phi(x,y)\right\} \vspace{10pt} \\
= \nabla^2 \phi(x,y) - \frac{a^2}{20} \left( \nabla^2 \right)^2 + O(a^6), \end{array}
\end{equation}
as one can verify by Taylor expansion. Thus ``stencil'' for discretizing a two dimensional Laplacian is illustrated in figure \ref{stencil}. An incidence matrix factorization for this stencil-based Laplacian is given in appendix \ref{stencil_incidence}. 

Theorem \ref{accumulation_lap} has the benefit that the error bound is characterized directly in terms of easily accessible quantities ($\phi$ and $\dot{\phi}$). However, the downside is that the condition on the error term (namely that it should take the form of a negative coefficient times the square of a Laplacian) is somewhat restrictive. Theorem \ref{accumulation} is more general in that the error term could be higher order, but still requires it to be the square of a differential operator. In appendix \ref{alt_smooth} we derive an alternative theorem which relaxes this restriction and can be applied to Laplacians that are constructed directly as a sum of discretized second partial derivatives. Relative to stencil-based discrete Laplacians such as in figure \ref{stencil} these Laplacians are much easier to derive and factor into incidence matrices at any order. On the other hand, we do not know how to use the methods of appendix \ref{alt_smooth} appears to obtain an error bound directly in terms of $\phi$ and $\dot{\phi}$. (In other words, appendix \ref{alt_smooth} contains only an analogue of theorem \ref{accumulation} but no analogue of theorem \ref{accumulation_lap}.) We include both versions of our analysis as we believe it may depend on context which one is more useful. A related question, which we leave for future work, is whether the specialized forms for the discretized Laplacians devised in this section and in appendix \ref{alt_smooth} result in smaller discretization errors than other discretized Laplacians at the same order. It is quite possible that they only aid in yielding provable error bounds but do not actually yield smaller error in practice.

\begin{figure}
	\begin{center}
		\includegraphics[width=0.45\textwidth]{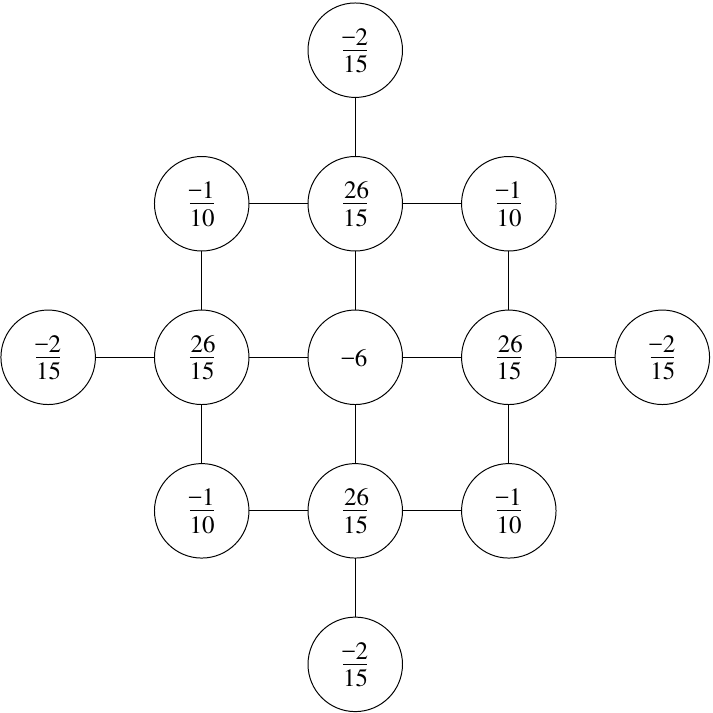}
	\end{center}
	\caption{\label{stencil} This linear combination of values at neighboring lattice sites produces a discrete approximation to the Laplacian with errors of order $a^2$ satisfying the conditions of theorem \ref{accumulation}. Specifically, one obtains $\nabla^2 \phi(x,y) - \frac{a^2}{20} \left( \nabla^2 \right)^2 + O(a^6)$. Thus the operator $\mathcal{D}$ in theorem \ref{accumulation} is in this case $\nabla^2$.}
\end{figure}

\section{Post-Processing}
\label{sec:postprocessing}

After performing Hamiltonian simulation we are left with a state which encodes both $\phi (T)$ and $B^{-1} d \phi (T) /dt$.
Depending on the application, we might be interested in just $\phi$ or just $d \phi /dt$ or both. 

If our goal is to produce a state proportional to $\phi$ then the post-processing amounts to measuring if the state is in $\mathcal{H}_V$ or $\mathcal{H}_E$ (recall the full Hilbert space is $\mathcal{H}_V \oplus \mathcal{H}_E$), with success if it is measured in $\mathcal{H}_V$. In general we cannot give a reasonable lower bound on the success probability of this measurement, even for simple systems. To see this, consider the case of the standing wave in 1D with Dirichlet boundary conditions. The initial conditions are $\phi (x,0)=\cos (x)$ and $d \phi (0) /dt=0$, and at any other time the field can be written $\phi (x,t)= f(t) \cos(x)$ for some $f$ that oscillates between 1 and -1. If the evolution time $T$ is chosen so that $f(T)=0$ then $\phi (x,T)=0$. So the state will have no support (up to errors from the finite difference method) in $\mathcal{H}_V$. However, at least in this example, for average choice of $T$ instead of worst-case, one will have an $O(1)$ probability of obtaining the $\phi$ subspace. The same issue arises if we are instead wish to extract $d \phi/dt$ from the complementary subspace.

If our goal is to produce a state proportional to $d \phi /dt$ then the post-processing is a little more complicated. We begin by measuring if the state is in $\mathcal{H}_V$ or $\mathcal{H}_E$, with success if it is measured in $\mathcal{H}_E$. The resulting state is proportional to $B^{-1} d \phi /dt$, so it remains to cancel $B^{-1}$. This inverse can be canceled in much the same way that $B^{-1}$ was originally applied. Mirroring the procedure for matrix inversion in \cite{harrow2009quantum}, the procedure for matrix multiplication is

\begin{align}
|B^{-1}d\phi/dt\rangle|0\rangle|0\rangle&=\sum_{j}\alpha_{j}|\Lambda_{j}\rangle|0\rangle|0\rangle\\&\mapsto\sum_{j}\alpha_{j}|\Lambda_{j}\rangle|\tilde{\lambda}_{j}\rangle|0\rangle\\&\mapsto\sum_{j}\alpha_{j}|\Lambda_{j}\rangle|\tilde{\lambda}_{j}\rangle\left(\frac{\tilde{\lambda}_{j}}{C}|0\rangle+\frac{\sqrt{C^{2}-\tilde{\lambda}_{j}^{2}}}{C}|1\rangle\right).\nonumber
\end{align}

The first line re-expresses the initial state in the eigenbasis $\{ | \Lambda_j \rangle \}$ of the Hamiltonian which is simulated in the subsequent phase estimation step. 

In the second line we run phase estimation on the unitary $\exp{(-iH)}$, where $H$ is exactly the same Hamiltonian we used for simulating the wave equation, and write the eigenvalues to the second register.
We use $| \Lambda_j \rangle$ to denote the eigenstate with eigenvalue $\lambda_j$, but we use $| \tilde \lambda_j \rangle$ to denote a state encoding the approximation of the eigenvalue output by phase estimation.

In the third line we perform a controlled rotation of the second qubit. The constant $C$ must satisfy $C \geq \sqrt{||L||}$ so that the argument under the square root is not negative. Setting it to $\Theta (\sqrt{||L||})$, the probability of measuring the last qubit in $| 0 \rangle$ is $\kappa (L)^{-2}$ in the worst case (i.e. when the initial state only has support in the ground space of the Hamiltonian.). Then we produce a state proportional to $d \phi (T) /dt$ conditioned on measuring the last qubit in the state $|0 \rangle$.

\section{Comparison to Other Quantum Algorithms}
\label{sec:comparison}

As discussed in the introduction, there are three quantum algorithms to which ours can be meaningfully compared. The algorithm of Clader, Jacobs, and Sprouse solves a problem related to, but not identical with, that solved here. Namely they give a quantum algorithm to compute scattering crossections in the special case of monochromatic illumination \cite{clader2013preconditioned}. In \cite{montanaro2016quantum}, Montanaro and Pallister analyze the degree to which quantum linear system algorithms can achieve speedups for finite element methods. The performance of such quantum algorithms when applied to wave equations is a complex question that we defer to future work.

The most direct comparison to our algorithm can be made with the algorithm of Berry, Childs, Ostrander, and Wang \cite{berry2017quantum}. Since the algorithm of \cite{berry2017quantum} only works for first order differential equations, we must introduce ancillary variables to simulate a second order differential equation. To simulate the wave equation for $\phi (x)$, we introduce the variable $\theta (x) \equiv a \frac{d \phi}{dt}$, in which case we have the first order equation

\begin{align}
\frac{d}{dt} \begin{bmatrix}
\phi \\ \theta
\end{bmatrix}
& = \frac{1}{a}
\begin{bmatrix}
0 & 1 \\
-L & 0 
\end{bmatrix}
\begin{bmatrix}
\phi \\ \theta
\end{bmatrix}
\end{align}

Let
\begin{equation}
\label{adef}
A = \frac{1}{a}
\begin{bmatrix}
0 & 1 \\
-L & 0 
\end{bmatrix}
\end{equation}
and let $V$ be a matrix that diagonalizes $A$:
\begin{equation}
A = V^{-1} D V \quad \quad \textrm{$D$ diagonal}.
\end{equation}
($V$ is defined only up to an overal normalization.) The complexity of the algorithm of \cite{berry2017quantum} is dictated by $\kappa_V$, the condition number of $V$ (which is independent of the normalization of $V$). Specifically, theorem 9 of \cite{berry2017quantum} gives a runtime upper bound for their quantum algorithm of
\begin{equation}
\label{complexity}
\widetilde{O} \left( \kappa_V s g T \| A \| \right),
\end{equation}
where $s$ is the sparsity of $A$, and $g$ is a measure of how much the norm of the solution vector $x(t)$ varies over the duration of the simulation, namely
\begin{equation}
g = \max_{t \in [0,T]} \| \vec{x}(t) \|/\| \vec{x}(T) \|.  
\end{equation}

We can see that for the problem at hand, as the lattice spacing $a$ is taken to zero:
\begin{eqnarray}
\label{scalings}
s & = & O(1) \nonumber \\
T & = & O(1) \nonumber \\
g & = & O(1) \nonumber \\
\|A \| & = & O(a^{-1}).
\end{eqnarray}
We can work out $\kappa_V$ by noting that $A$ is diagonalized by the matrix whose columns are the eigenvectors of $A$. That is, if the eigenvectors of $A$ are $\vec{v}_1,\ldots,\vec{v}_N$ with corresponding eigenvalues $\lambda_1,\ldots,\lambda_N$ then $V^{-1} A V = \textrm{diag}(\lambda_1,\ldots,\lambda_N)$ where
\begin{equation}
V = 
\left[
\begin{array}{cccc}
\vrule & \vrule & & \vrule\\
\vec{v}_{1} & \vec{v}_{2} & \ldots & \vec{v}_N \\
\vrule & \vrule & & \vrule 
\end{array}
\right]
\end{equation}
Let $\vec{y}_1,\ldots,\vec{y}_N$ denote the eigenvectors of $L$. By inspecting \eq{adef} one sees that the eigenvectors of $A$ are
\begin{equation}
\left[ \begin{array}{c} \vec{y}_1 \\ i \sqrt{\lambda_1} \vec{y}_1 \end{array} \right],
\left[ \begin{array}{c} \vec{y}_1 \\ -i \sqrt{\lambda_1} \vec{y}_1 \end{array} \right],
\ldots,
\left[ \begin{array}{c} \vec{y}_M\\ i \sqrt{\lambda_M} \vec{y}_M \end{array} \right],
\left[ \begin{array}{c} \vec{y}_M \\ -i \sqrt{\lambda_M} \vec{y}_M \end{array} \right].
\end{equation}
($M$ is the dimension of $L$ and $N = 2M$ is the dimension of $A$.)

We can thus write $V$ in the following block form.
\begin{equation}
V = \left[ \begin{array}{cc}
Y    & Y \\
i Z & -i Z
\end{array}
\right]
\end{equation}
where
\begin{equation}
Y = 
\left[
\begin{array}{cccc}
\vrule & \vrule & & \vrule\\
\vec{y}_{1} & \vec{y}_{2} & \ldots & \vec{y}_M \\
\vrule & \vrule & & \vrule 
\end{array}
\right]
\end{equation}
and
\begin{equation}
Z = 
\left[
\begin{array}{cccc}
\vrule & \vrule & & \vrule\\
\sqrt{\lambda_1} \vec{y}_{1} & \sqrt{\lambda_2} \vec{y}_{2} & \ldots & \sqrt{\lambda_M} \vec{y}_M \\
\vrule & \vrule & & \vrule 
\end{array}
\right].
\end{equation}
$L$ is a symmetric matrix so $\vec{y}_1,\ldots,\vec{y}_M$ form an orthogonal basis. We choose the normalizations to make it orthonormal. Let $U$ be the orthogonal matrix that diagonlizes $Y$. Then
\begin{equation}
\left[ \begin{array}{cc} U^T & 0 \\ 0 & U^T \end{array} \right]
\left[ V \right]
\left[ \begin{array}{cc} U & 0 \\ 0 & U \end{array} \right] =
\left[ \begin{array}{cc} \id & \id \\ iS & -iS \end{array} \right],
\end{equation}
where
\begin{equation}
S = \left[ \begin{array}{ccc} \sqrt{\lambda_1} & & \\ & \ddots & \\ & & \sqrt{\lambda_M} \end{array} \right].
\end{equation}
Permuting the basis then yields
\begin{equation}
\left[ \begin{array}{cccc} B_1 & & & \\
& B_2 & & \\
& & \ddots & \\
& & & B_M
\end{array} \right]
\end{equation}
where for each $j=1,\ldots,M$ the block $B_j$ is given by the following $2 \times 2$ matrix
\begin{equation}
B_j = \left[ \begin{array}{cc} 1 & 1 \\ i \sqrt{\lambda_j} & -i \sqrt{\lambda_j} \end{array} \right].
\end{equation}
This preceeding manipulations were all changes of basis, which do not affect the eigenspectrum of. Thus, the eigenvalues of $V$ are the eigenvalues of $B_1,\ldots,B_M$. By direct calculation, the eigenvalues of $B_j$ are $q^{(+)}_j$ and $q^{(-)}_j$ where
\begin{equation}
q^{(\pm)}_j = \frac{1}{2} \left( 1 + i \sqrt{\lambda_j} \pm \sqrt{1-6 i \sqrt{\lambda_j} - \lambda_j} \right).
\end{equation}

For a path graph of $N$ vertices the eigenvalues of the Laplacian range from $\sim 1/N^2$ to $1$, and the same is true for any larger constant number of dimensions for the eigenvalues of an $N \times N \times \ldots \times N$ grid. The smallest eigenvalue of $V$ is thus $q_i^{(-)}$ with where $i$ indexes the smallest eigenvalue of $L$. Thus, for large $N$, we can approximate $q_i^{(-)}$ by Taylor expanding to lowest order in $\sqrt{\lambda_i}$, obtaining
\begin{eqnarray}
q_i^{(-)} & = & \frac{1}{2} \left( 1 + i \sqrt{\lambda_j} - \sqrt{1-6 i \sqrt{\lambda_j} - \lambda_j} \right) \\
& \simeq & 2 i \sqrt{\lambda_j},
\end{eqnarray}
which is of order $a$. Similarly, we can see that the largest eigenvalue of $V$ is $O(1)$ and thus 
\begin{equation}
\label{kappav}
\kappa_V = \Theta(a^{-1}).
\end{equation}
Substituting \eq{kappav} and \eq{scalings} into \eq{complexity} yields a total complexity of $O(a^{-2})$ for the quantum algorithm of \cite{berry2017quantum}.

In the algorithm presented here, we have quadratically better dependence on $\kappa$. There are three places for this dependence to come into the total complexity of our algorithm. First, if we choose to prepare an arbitrary initial state, then the first step of our algorithm is to implement, via quantum linear algebra methods \cite{harrow2009quantum, childs2015quantum,ambainis} the Moore-Penrose pseudoinverse of the incidence matrix $B$. The complexity of this step is proportional to the condition number of $B$, which is the square root of the condition number of the Laplacian $L$ \footnote{The condition numbers of $B$ and $L$ will both depend on the connectivity of the lattice. These condition numbers can be large if scatterers are present which create bottlenecks in the lattice, i.e., convex locations where only a few edges can be removed that will partition the lattice into two relatively large components.}. A second place that the condition number can contribute to the complexity is in the post-processing, as we saw when we considered producing a state proportional to $d \phi /dt$. Here our approach also scales quadratically better with respect to the condition number of the Laplacian. Additionally, the number of qubits required by our algorithm is $\log N$ where $N$ is the number of lattice sites, whereas the number of qubits required by the algorithm of \cite{berry2017quantum} is $O(\log(N)+ \log t)$, where $t$ is the duration of the process to be simulated.

It is worthwhile to relate the Laplacian's condition number, which is a fairly abstract quantity, with parameters of more direct physical significance. In the case of a Laplacian for a $D$-dimensional cubic volume of dimension $\ell \times \ell \times \ldots$ discretized into a cubic lattice of spacing $a$ one sees that the largest eigenvalue of $-\frac{1}{a^2} L$ is of order $D/a^2$ and the smallest eigenvalue is of order $1/\ell^2$. Thus the condition number of the Laplacian is of order $D\ell^2/a^2$, so the incidence matrix has a condition number of order $\sqrt{D} \ell / a$. In our algorithm, the simulation of the time-evolution itself, achieved using \cite{berry2015hamiltonian}, scales as $\widetilde{O}(stD/a)$. Thus, both state preparation and time-evolution have complexity scaling linearly in $a^{-1}$.

\section{Klein-Gordon Equation}
\label{sec:KG}

Going to relativistic theories we know that spinless particles are described by the Klein-Gordon equation,
\begin{equation}
\label{Klein_Gordon}
\frac{1}{c^{2}}\frac{\partial^{2}\phi}{\partial t^{2}}-\nabla^{2}\phi+\frac{m^{2}c^{2}}{\hbar^{2}}\phi=0,
\end{equation}
where $m$ is the particle mass, $c$ is the speed of light and $\hbar$ the Planck constant. In order to not carry these constants any more we will adopt the natural units, which implies $c=1$ and $\hbar=1$. 

As we can see we are dealing with a wave equation, and thus it also should admit some Hamiltonian in our Schr\"{o}dinger equation. Suppose we have a graph $G'$, where 
\begin{equation*}
\frac{\partial^{2}\phi}{\partial t^{2}}=\frac{1}{a^{2}}L\left(G'\right)\phi,
\end{equation*}
is the discretized version of Eq.(\ref{Klein_Gordon}). It means that our Laplacian has the whole information about the particle, which includes its mass term. In fact this graph $G'$ can be  easily achieved from a graph $G$ that gives our ordinary wave equation, which means $L\left(G\right)$ does not have a mass term.

Starting with $G$ the mass term can be realized by adding self loops with {$W=\left(am\right)^2$} as its weight on all vertices of $G$. This manipulated graph is our graph $G'$. Finally, as we did before, we need to construct its incidence matrix $B\left(G'\right)$ in order to get the Laplacian, 
\[
B\left(G'\right)^{\dagger}B\left(G'\right)=L\left(G'\right).
\]
Besides, without difficult we can see how this Laplacian is related with the Laplacian from $G$
\[
L\left(G'\right)=L\left(G\right)+{a^2}m^{2}I,
\]
where $I$ is the identity matrix. Therefore, whereas  $B\left(G\right)$ gives our ordinary wave equation, applying $B\left(G'\right)$  in our Hamiltonian gives our relativistic wave equation.

\section{Maxwell's Equations}
\label{sec:Maxwell}

With $\mu_0=\epsilon_0=1$ and without sources, Maxwell's equations governing the time evolution of electric and magnetic fields take the form

\[
\frac{\partial \vec E}{\partial t} =  \vec \nabla \times \vec B  \ \ \
\frac{\partial \vec B}{\partial t} =  - \vec \nabla \times \vec E 
\]

\noindent which imply that $\vec E$ and $\vec B$ both follow the wave equation. If we consider discretizing space, then we can write these as

\[
\frac{\partial }{\partial t}
\begin{bmatrix}
\vec E \\
\vec B
\end{bmatrix}
= 
\begin{bmatrix}
0 & C \\
- C & 0 
\end{bmatrix}
\begin{bmatrix}
\vec E \\
\vec B
\end{bmatrix}
\]

\noindent where $C$ is the finite difference approximation of the curl operator. To see how to construct $C$, consider the following
\[
\vec \nabla \times  \begin{bmatrix} a \\ b \\ c\end{bmatrix}  =  \begin{bmatrix}
\partial c / \partial y - \partial b / \partial z \\ 
\partial a / \partial z - \partial c / \partial x \\ 
\partial b / \partial x - \partial a / \partial y \\ 
\end{bmatrix} 
=  \begin{bmatrix}
0 & - \partial / \partial z & \partial / \partial y \\
\partial / \partial z & 0 & - \partial / \partial x \\
-\partial / \partial y & \partial / \partial x & 0
\end{bmatrix}
\begin{bmatrix}
a\\ b \\ c
\end{bmatrix}.
\]

This suggests we should consider the linear differential equation
\begin{equation}
\frac{\partial }{\partial t} \begin{bmatrix}
E_x \\ E_y \\ E_z \\ B_x \\ B_y \\ B_z
\end{bmatrix}= \begin{bmatrix}
0 & 0 & 0 & 0 & - \partial / \partial z & \partial / \partial y \\
0 & 0 & 0 & \partial / \partial z & 0 & - \partial / \partial x \\
0 & 0 & 0 & -\partial / \partial y & \partial / \partial x & 0 \\
0 &  \partial / \partial z & -\partial / \partial y & 0 & 0 & 0 \\
-\partial / \partial z & 0 &  \partial / \partial x & 0 & 0 & 0 \\
\partial / \partial y & -\partial / \partial x & 0 & 0 & 0 & 0
\end{bmatrix}
\begin{bmatrix}
E_x \\ E_y \\ E_z \\ B_x \\ B_y \\ B_z
\end{bmatrix}
\label{Maxwell}
\end{equation}
We can discretize space into a uniform cubic lattice and approximate the differential operators using finite difference methods to reduce this to an ordinary differential equation. (Appendix \ref{ordertables} contains numerical values for the entries of these operators up to tenth order.)
This ordinary differential equation will be a case of Schr\"{o}dinger's equation since the approximate differential operators coming from the Lagrange interpolation formula are anti-Hermitian. In this case, unitarity translates to conservation of the classical energy contained in the field $\displaystyle\int_V | \vec E (\vec x )|^2 + | \vec B (\vec x )|^2$.

\section{Future Work}

It is an interesting open question whether our quantum algorithm is optimal. In particular, it is natural to ask whether an analogue of the no-fast-forwarding theorem from \cite{BACS07} could yield a lower bound for the complexity of the problem of simulating wave equations that matches the complexity of the algorithm presented here. It is also interesting to investigate the performance of quantum algorithms for simulating the wave equation based on finite element methods, rather than finite difference methods, as considered here. Another direction for future work is to use automated circuit synthesis techniques to generate concrete quantum circuits implementing our algorithm and thereby obtain quantitative resource estimates for benchmark instances of wave equation simulation problems. Lastly, one can consider extending the quantum algorithm presented here to more complicated wave equations.


\section{Acknowledgements}

The authors thank Yi-Kai Liu and Eite Tiesinga for insightful discussions. The authors also thank David Gosset, Gorjan Alagic, Peter Bierhorst, and anonymous referees for useful feedback on the manuscript. This research was supported by the Department of Energy under award number DE-SC0016431. Parts of this research were completed while SJ was an employee of the National Institute of Standards and Technology, an agency of the US government. The resulting portions of this manuscript are not subject to US copyright.

\appendix

\section{Alternative Smoothness Analysis}
\label{alt_smooth}

\begin{theorem}
	\label{alt_accumulation}
	Let $\phi_{\lambda}$ be the solution to
	\begin{equation}
	\label{alt_lambdawave}
	\frac{\partial^2 \phi_\lambda}{\partial t^2} = \nabla^2 \phi_{\lambda} + \lambda^2 \sum_{j=1}^d \left( \frac{ \partial^k}{\partial x_j^k} \right)^2 \phi_{\lambda}
	\end{equation}
	on some compact continuous domain $\Omega \subset \mathbb{R}^d$ subject to some specified boundary conditions. We take initial conditions at $t=0$ to be fixed functions $\phi(\vec{x},0)$ and $\dot{\phi}(\vec{x},0)$ independent of $\lambda$. Then for any $\epsilon \in \mathbb{R}$ and any $t \geq 0$
	\begin{equation*}
	\|\phi_{\epsilon}(t)-\phi_{0}(t)\|\leq\sqrt{2t\epsilon}\left[\left(\|\phi(0)\|^{2}+\sum_{j=1}^{d}\|\psi_{j}(0)\|^{2}\right)\left(\sum_{j=1}^{d}\left(\left\Vert \frac{\partial}{\partial x_{j}^{k}}\phi(0)\right\Vert ^{2}+\sum_{l=1}^{d}\left\Vert \frac{\partial^{k}}{\partial x_{j}^{k}}\psi_{l}(0)\right\Vert ^{2}\right)\right)\right]^{1/4}.
	\end{equation*}
	where $\| f \| \equiv \sqrt{\int_\Omega d^d x |f(\vec{x})|^2}$ and
	\begin{equation}
	\vec{\psi}(\vec{x}, 0) = \int d^d y \frac{\vec{x}-\vec{y}}{|\vec{x} - \vec{y}|^d} \dot{\phi}(\vec{y},0).
	\end{equation}
\end{theorem}

\begin{proof}
	Let
	\begin{eqnarray}
	\label{alt_slambdadef}
	S_\lambda & = & \left[ \begin{array}{c} \phi_\lambda \\ \vec{\psi}_\lambda \\ \vec{\theta}_{\lambda} \end{array} \right] \label{alt_slambda}\\
	\nabla & = & \left[ \frac{\partial}{\partial x_1}, \ldots, \frac{\partial}{\partial x_d} \right] \\
	\nabla_k & = & \left[ \frac{\partial^k}{\partial x_1^k}, \ldots, \frac{\partial^k}{\partial x_d^k} \right] \\
	H_0 & = & \left[ \begin{array}{ccc}
	0 & \nabla & 0 \\
	-\nabla^T  & 0 & 0 \\
	0 & 0 & 0
	\end{array} \right]\\
	H_1 & = & \left[ \begin{array}{ccc}
	0 & 0 & \nabla_k \\
	0 & 0 & 0 \\
	-\nabla_k^T & 0 & 0 \end{array} \right] \\
	H_\lambda &  = & H_0 + \lambda H_1 \\
	\frac{dS_{\lambda}}{dt} & = & -i H_\lambda S_\lambda. \label{alt_fake_Schrod}
	\end{eqnarray}
	By \eq{alt_fake_Schrod},
	\begin{eqnarray}
	\frac{d^2}{dt^2} S_\lambda & = & - H_\lambda^2 S_\lambda \\
	& = & \left[ \begin{array}{ccc}
	\nabla^2 + \lambda^2 \nabla_k^2 & 0 & 0 \\
	0 & \nabla^T \nabla & \lambda \nabla^T \nabla_k \\
	0 & \lambda \nabla^T \nabla_k & \lambda^2 \nabla_k^T \nabla_k
	\end{array} \right] \left[ \begin{array}{c} \phi \\ \vec{\psi} \\ \vec{\theta} \end{array} \right]. 
	\end{eqnarray}
	Thus the solution to \eq{alt_fake_Schrod} satisfies \eq{alt_lambdawave}. As initial conditions $(t=0)$ for $\vec{\psi}_\lambda$ we take
	\begin{eqnarray}
	\vec{\psi}(0) & = & \textrm{Div}^{-1} \left[ \dot{\phi}(0) \right] \label{init_psi} \\
	\vec{\theta}(0) & = & 0 \label{alt_init_theta}
	\end{eqnarray}
	where $\textrm{Div}^{-1}$ is as defined in \eq{divGreens}. 
	By \eq{alt_fake_Schrod} we have
	\begin{eqnarray}
	\frac{d}{dt}\langle S_{\epsilon},S_{0}\rangle&=&\left\langle \dot{S}_{\epsilon}(t),S_{0}(t)\right\rangle +\left\langle S_{\epsilon}(t),\dot{S}_{0}(t)\right\rangle \\
	&=&\langle-i(H_{0}+\epsilon H_{1})S_{\epsilon}(t),S_{0}(t)\rangle+\langle S_{\epsilon}(t),-iH_{0}S_{0}(t)\rangle\\
	&=&\langle S_{\epsilon}(t),i(H_{0}+\epsilon H_{1})S_{0}(t)\rangle\langle S_{\epsilon}(t),-iH_{0}S_{0}(t)\rangle\\
	&=&i\epsilon\langle S_{\epsilon}(t),H_{1}S_{0}(t)\rangle.
	\end{eqnarray}
	Thus, by the Cauchy-Schwarz inequality
	\begin{equation}
	\label{alt_unsimp}
	\left| \frac{d}{dt} \langle S_\epsilon, S_0 \rangle \right| \leq \epsilon \| S_\epsilon(t) \| \times \| H_1 S_0(t) \|,
	\end{equation}
	where $\| S \|$ is a shorthand for $\sqrt{ \langle S, S \rangle }$. $H_\lambda$ is Hermitian for real $\lambda$ and therefore $\| S_\epsilon(t) \| = \| S_\epsilon(0) \|$. Thus \eq{alt_unsimp} simplifies to
	\begin{equation}
	\label{alt_chbound1}
	\left| \frac{d}{dt} \langle S_\epsilon, S_0 \rangle \right| \leq \epsilon \| S_\epsilon(0) \|  \times \| H_1 S_0(t) \|.
	\end{equation}
	Next, observe that
	\begin{equation}
	\label{alt_observation}
	\| H_1 S_0(t) \| = \sqrt{ \left\| \sum_{j=1}^d \frac{\partial^k}{\partial x_j^k} \theta_j \right\|^2 + \sum_{j=1}^d \left\| \frac{\partial^k}{\partial x_j^k} \phi \right\|^2}.
	\end{equation}
	and
	\begin{eqnarray}
	\left\| \frac{\partial^k}{\partial x_j^k} \phi \right\|^2 & \leq & \left\| \frac{\partial^k}{\partial x_j^k} \phi \right\|^2 + \sum_{l=1}^d \left\| \frac{\partial^k}{\partial x_j^k} \psi_l \right\|^2 \\
	& = & \left\| \mathcal{H}_j^{(k)} S_0(t) \right\|^2
	\end{eqnarray}
	where
	\begin{equation}
	\label{alt_hjkdef}
	\mathcal{H}_j^{(k)} = i^k \left[ \begin{array}{c|ccc|ccc}
	\frac{\partial^k}{\partial x_j^k} & & & & & & \\
	\hline
	&  \frac{\partial^k}{\partial x_j^k} & & & & & \\
	& & \ddots & & & & \\
	& & & \frac{\partial^k}{\partial x_j^k} & & & \\
	\hline
	& & & & 0 & & \\
	& & & & & \ddots & \\
	& & & & & & 0 \end{array} \right].
	\end{equation}
	For any $k$, $\mathcal{H}_j^{(k)}$ commutes with the $H_0$ and is Hermitian. Thus
	\begin{equation}
	\left\| \mathcal{H}_j^{(k)} S_0(t) \right\| = \left\| \mathcal{H}_j^{(k)} S_0(0) \right\|.
	\end{equation}
	Next, we observe that
	\begin{eqnarray}
	\left\| \sum_{j=1}^d \frac{\partial^k}{\partial x_j^k} \theta_j \right\| & \leq & \sum_{j=1}^d \left\| \frac{\partial^k}{\partial x_j^k} \theta_j \right\| \\
	& \leq & \sqrt{d} \sqrt{ \sum_{j=1}^d \left\| \frac{\partial^k}{\partial x_j^k} \theta_j \right\|^2} \\
	& = & \sqrt{d} \| \mathcal{H}_\theta S_0(t) \|
	\end{eqnarray}
	where
	\begin{equation}
	\mathcal{H}_\theta = \left[ \begin{array}{c|ccc|ccc}
	0 & 0 & \ldots & 0 & 0 & \ldots & 0 \\
	\hline
	0 & 0 & \ldots & 0 & 0 & \ldots & 0 \\
	\vdots & \vdots & & \vdots & \vdots & & \vdots \\
	0 & 0 & \ldots & 0 & 0 & \ldots & 0 \\
	\hline
	0 & 0 & \ldots & 0 & \frac{\partial^k}{\partial x_1^k} &  &  \\
	\vdots & \vdots & & \vdots & & \ddots & \\
	0 & 0 & \ldots & 0 & & & \frac{\partial}{\partial x_d^k}
	\end{array} \right].
	\end{equation}
	$\mathcal{H}_\theta$ is Hermitian and commutes with $H_0$ thus, by \eq{alt_init_theta},
	\begin{equation}
	\| \mathcal{H}_\theta S_0(t) \| = \| \mathcal{H}_\theta S_0(0) \| = 0.
	\end{equation}
	Substituting these results into \eq{alt_observation} yields
	\begin{equation}
	\label{alt_H1bound}
	\| H_1 S_0(t) \| \leq \sqrt{\sum_{j=1}^d \left\| \mathcal{H}_j^{(k)} S_0(0) \right\|^2}.
	\end{equation}
	Substituting \eq{alt_H1bound} into \eq{alt_chbound1} yields
	\begin{equation}
	\label{alt_changebound}
	\left| \frac{d}{dt} \langle S_\epsilon, S_0 \rangle \right| \leq \epsilon \| S_\epsilon(0) \|  \times \sqrt{\sum_{j=1}^d \left\| \mathcal{H}_j^{(k)} S_0(0) \right\|^2}.
	\end{equation}
	By definition
	\begin{eqnarray}
	\| S_\epsilon(t) - S_0(t) \|^2 &=& \langle S_\epsilon(t) - S_0(t), S_\epsilon(t) - S_0(t) \rangle \nonumber\\
	&=&\langle S_\epsilon(t), S_\epsilon(t) \rangle + \langle S_0(t), S_0(t) \rangle \nonumber \\
	&-&2 \mathrm{Re} \langle S_\epsilon, S_0 \rangle.
	\end{eqnarray}
	The ``Hamiltonians'' $H_0$ and $H_1$ are Hermitian so $\langle S_\epsilon(t), S_\epsilon(t) \rangle$ and  $\langle S_0(t), S_0(t) \rangle$ are time-independent for any $\epsilon \in \mathbb{R}$. Thus,
	\begin{equation}
	\frac{d}{dt} \| S_\epsilon(t) - S_0(t)\|^2 =  -2 \mathrm{Re} \frac{d}{dt} \langle S_\epsilon, S_0 \rangle.
	\end{equation}
	Thus, by \eq{alt_changebound}
	\begin{equation}
	\label{alt_changebound2}
	\left| \frac{d}{dt} \| S_\epsilon(t) - S_0(t)\|^2 \right| \leq 2 \epsilon \| S_\epsilon(0) \|  \times \sqrt{\sum_{j=1}^d \left\| \mathcal{H}_j^{(k)} S_0(0) \right\|^2}.
	\end{equation}
	By the triangle inequality
	\begin{eqnarray}
	\| S_\epsilon(t) - S_0(t) \|^2 & = & \| S_\epsilon(0) - S_0(0) \|^2\\ 
	&+& \int_0^t d\tau \frac{d}{d\tau} \| S_\epsilon(\tau) - S_0(\tau) \|^2 \nonumber\\
	& \leq &  \| S_\epsilon(0) - S_0(0) \|^2 \nonumber\\
	&+& \int_0^t d\tau \left| \frac{d}{d\tau} \| S_\epsilon(\tau) - S_0(\tau) \|^2 \right|.\nonumber
	\end{eqnarray}
	The initial conditions have $S_\epsilon(0) = S_0(0)$, and therefore
	\begin{equation}
	\label{alt_Sdiff}
	\| S_\epsilon(t) - S_0(t) \|^2 \leq \int_0^t d\tau \left| \frac{d}{d\tau} \| S_\epsilon(\tau) - S_0(\tau) \|^2 \right|.
	\end{equation}
	Applying \eq{alt_changebound2} to \eq{alt_Sdiff} yields
	\begin{equation}
	\| S_\epsilon(t) - S_0(t) \|^2 \leq 2 t \epsilon \| S_\epsilon(0) \|  \times \sqrt{\sum_{j=1}^d \left\| \mathcal{H}_j^{(k)} S_0(0) \right\|^2}.
	\end{equation}
	Recalling the definition of $S_\lambda$ \eq{alt_slambda},
	\begin{equation}
	\| S_\epsilon(t) - S_0(t) \|^2 = \| \phi_\epsilon(t) - \phi_0(t) \|^2 + \| \vec{\psi}_\epsilon(t) - \vec{\psi}_0(t) \|^2. 
	\end{equation}
	Thus \eq{alt_Sdiff} implies the bound
	\begin{equation}
	\label{alt_almostdone}
	\| \phi_\epsilon(t) - \phi_0(t) \|^2 \leq 2 t \epsilon \| S_\epsilon(0) \|  \times \sqrt{\sum_{j=1}^d \left\| \mathcal{H}_j^{(k)} S_0(0) \right\|^2}.
	\end{equation}
	By \eq{alt_hjkdef}, \eq{alt_slambdadef}, and \eq{alt_init_theta}, \eq{alt_almostdone} becomes
	\begin{equation}
	\|\phi_{\epsilon}(t)-\phi_{0}(t)\|^{2}\leq2t\epsilon\sqrt{\left(\|\phi(0)\|^{2}+\sum_{j=1}^{d}\|\psi_{j}(0)\|^{2}\right)}\sqrt{\left(\sum_{j=1}^{d}\left(\left\Vert \frac{\partial}{\partial x_{j}^{k}}\phi(0)\right\Vert ^{2}+\sum_{l=1}^{d}\left\Vert \frac{\partial^{k}}{\partial x_{j}^{k}}\psi_{l}(0)\right\Vert ^{2}\right)\right).}\nonumber
	\end{equation}
\end{proof}
Theorem \ref{alt_accumulation} gives a very nice quantitative upper bound on discretization errors in terms of directly accessible properties of the initial conditions. Furthermore, theorem \ref{blowup_lemma} shows that the quantity $\vec{\psi}(0)$ has magnitude not too much larger than the chosen initial velocity $\dot{\phi}(0)$. However, theorem \ref{alt_accumulation} applies only under the specific condition that the error term of interest is expressible as a positive coefficient times the sum of $(2k)\th$ derivatives. Not all discretized Laplacians satisfy this. However, it is possible to engineer high order Laplacians such that this is the case. This problems reduces to engineering a high order discretized one-dimensional derivatives such that the leading error term is a positive coefficient times an even derivative. The Laplacian in $d$ dimensions can then be composed as the sum of these discretized derivatives along each of the coordinate axes. 

We illustrate this by giving an explicit discretized Laplacian in one dimension with error of order $a^4$, which satisfies this condition and then computing a corresponding incidence matrix factorization. By Taylor expansion, one can verify that
\begin{eqnarray}
&  & -\frac{9}{2} f(x) + \frac{17}{6} \left( f(x+a) + f(x-a) \right)  \nonumber\\
& & - \frac{41}{60} \left( f(x+2a) + f(x-2a) \right) + \frac{1}{10} \left( f(x+3a) + f(x-3a) \right) \nonumber\\
& = & a^2 \frac{d^2 f}{dx^2}(x) + \frac{4}{45} a^6 \frac{d^6 f}{dx^6}(x) + O(a^8).
\end{eqnarray}
On a one dimensional lattice with periodic boundary conditions we can write this Laplacian as
\[
L^{(4)} = a_0 \id + a_1 (S + S^-1) + a_2 (S^2 + S^{-2})+ a_3 (S^3 + S^{-3})
\]
where $S$ is the cyclic shift operator and
\begin{eqnarray*}
	a_0 & = & -9/2 \\
	a_1 & = & 17/6 \\
	a_2 & = & -41/60 \\
	a_3 & = & 1/10.
\end{eqnarray*}
Next, we verify that this can be factorized as
\begin{equation}
\label{alt_factorization}
L^{(4)} = -B^T B
\end{equation}
with sparse $B$. To this end we introduce the ansatz
\begin{equation}
B = b_0 \id + b_1 S + b_2 S^2 + b_3 S^3.
\end{equation}
The requirement \eq{alt_factorization} then determines a system of quadratic equations constraining $b_0,b_1,b_2,b_3$. One solution to this system of equations is (up to 6 digits of precision)
\begin{eqnarray*}
	b_0 & = & \phantom{-}1.27811 \\
	b_1 & = & -1.63446 \\
	b_2 & = & \phantom{-}0.434589 \\
	b_3 & = & -0.0782406 
\end{eqnarray*}
as one can verify.

\section{Analytical Q} \label{Q_statics}

We begin by giving the mesh spacing as a function of the number of vertices $\left|V\right|=n$ for our one dimensional lattice
\begin{equation}
a\left(n\right)=\frac{1}{n+1}.
\end{equation}

As discussed in \sect{sec:errors}, in order to get $Q$ we need to work with the three different mesh spacing $a_{1}$, $a_{2}$ and $a_{3}$, where the relation between them can be established  working with the follow total number of vertices
\begin{eqnarray}
a_{1}\left(4n+3\right)&=&\frac{1}{4\left(n+1\right)},\\a_{2}\left(2n+1\right)&=&\frac{1}{2\left(n+1\right)},\nonumber\\a_{3}\left(n\right)&=&\frac{1}{n+1}\nonumber,
\end{eqnarray}
respectively. Moving forward we get three discrete functions that describe the standing wave,
\begin{eqnarray}
\phi_{j}^{a}&=&\cos\left(\omega^{a}t\right)\sin\left(\frac{\pi}{4\left(n+1\right)}j\right),\\\phi_{j}^{2a}&=&\cos\left(\omega^{2a}t\right)\sin\left(\frac{\pi}{2\left(n+1\right)}j\right),\nonumber\\\phi_{j}^{4a}&=&\cos\left(\omega^{4a}t\right)\sin\left(\frac{\pi}{n+1}j\right)\nonumber
\end{eqnarray}
where $\omega$ is the frequency of the wave,

\begin{eqnarray}
\omega^{a}&=&8\left(n+1\right)\sin\left(\frac{\pi}{8\left(n+1\right)}\right),\\\omega^{2a}&=&4\left(n+1\right)\sin\left(\frac{\pi}{4\left(n+1\right)}\right),\nonumber\\\omega^{4a}&=&2\left(n+1\right)\sin\left(\frac{\pi}{2\left(n+1\right)}\right)\nonumber.
\end{eqnarray}
From the $Q$ factor definition we know that we need to compute two differences $\varPhi^{4a}-\varPhi^{2a}$ and $\varPhi^{2a}-\varPhi^{a}$. However, these points should be computed at the same distance, which means $\varPhi_{j}^{4a}-\varPhi_{2j}^{2a}$, and $\varPhi_{2j}^{2a}-\varPhi_{4j}^{a}$. Let us proceed with the follow computation,
\begin{equation*}
\varPhi_{j}^{4a}-\varPhi_{2j}^{2a}=\left(\cos\left(\omega^{4a}t\right)-\cos\left(\omega^{2a}t\right)\right)\sin\left(\frac{\pi}{n+1}j\right).
\end{equation*}
But we are interest in the continuum limit of this expression, with means $a\rightarrow0$ or $n\rightarrow\infty$. Thus, from now the idea is to work with approximate values. Starting with the frequency,
\begin{eqnarray*}
	\omega^{4a}&\simeq&\pi-\delta_{4a},\\\omega^{2a}&\simeq&\pi-\delta_{2a}.
\end{eqnarray*}
where
\[
\delta_{4a}=-\frac{\pi^{3}}{24\left(n+1\right)^{2}},
\]
and
\[
\delta_{2a}=-\frac{\pi^{3}}{96\left(n+1\right)^{2}}.
\]
Now we can use the following trigonometric property,
\[
\cos\left[\left(\pi-\delta_{4a}\right)t\right]-\cos\left[\left(\pi-\delta_{2a}\right)t\right]=-2\sin\left(\bar{\omega}t\right)\sin\left(\delta t\right),
\]
with 
\begin{eqnarray*}
	\bar{\omega}&=&\pi-\frac{\delta_{4a}-\delta_{2a}}{2},\\\delta&=&\frac{\delta_{4a}-\delta_{2a}}{2}.
\end{eqnarray*}
But for large $n$ we get the follow approximations
\begin{eqnarray*}
	\sin\left(\bar{\omega}t\right)&\simeq&\sin\left(\pi t\right),\\\sin\left(\delta t\right)&\simeq&-\frac{3\pi^{3}}{192n^{2}}.
\end{eqnarray*}
However, our real interest is computing the norm $\left\Vert \varPhi^{4a}-\varPhi^{2a}\right\Vert _{2}$ in the continuum limit,
{\footnotesize
	\begin{eqnarray*}
		\left\Vert \varPhi^{4a}-\varPhi^{2a}\right\Vert _{2}&=&\lim_{n\rightarrow\infty}\sqrt{\frac{1}{n}\sum_{j=1}^{n}\left(\varPhi_{j}^{4a}-\varPhi_{2j}^{2a}\right)^{2}},\\&=&\lim_{n\rightarrow\infty}\sqrt{\frac{1}{n}\sum_{j=1}^{n}4\sin^{2}\left(\pi t\right)\left(\frac{3\pi^{3}}{192n^{2}}\right)^{2}\sin\left(\frac{\pi}{n+1}j\right)}.
\end{eqnarray*}}
where we can make use of the expression below,
\begin{equation*}
\lim_{n\rightarrow\infty}\frac{1}{n}\sum_{j=0}^{n}\sin^{2}\left(\frac{\pi j}{n}\right)=\int_{0}^{1}dx\sin^{2}\left(\pi x\right)\,=\,\frac{1}{2},
\end{equation*}
Therefore,
\begin{equation*}
\left\Vert \varPhi^{4a}-\varPhi^{2a}\right\Vert _{2}=\sqrt{2}\sin\left(\pi t\right)\left(\frac{3\pi^{3}}{192n^{2}}\right).
\end{equation*}
Similarly, for $\left\Vert \varPhi^{2a}-\varPhi^{a}\right\Vert _{2}$ we get
\begin{eqnarray*}
	\left\Vert \varPhi^{2a}-\varPhi^{a}\right\Vert _{2}&=&\lim_{n\rightarrow\infty}\sqrt{\frac{1}{n}\sum_{j=1}^{n}\left(\varPhi_{2j}^{2a}-\varPhi_{4j}^{a}\right)^{2}},\\&=&\frac{1}{4}\sqrt{2}\sin\left(\pi t\right)\left(\frac{3\pi^{3}}{192n^{2}}\right).
\end{eqnarray*}
Thus, combining these two results in the Q factor expression we establish
\begin{equation*}
\underset{a\rightarrow0}{Q\left(t\right)}=4,
\end{equation*}
that agrees with the value for $e_{2}$ in the Richardson expansion and with our numerical result.

The same steps can be done for the second order Laplacian to see $Q\left(t\right)=16$ in the continuum limit. However the correct wave frequency for this case is 
\[
\omega=\left(n+1\right)\sqrt{\frac{5}{2}-\frac{8}{3}\cos\left(\frac{\pi}{n+1}\right)+\frac{1}{6}\cos\left(\frac{2\pi}{n+1}\right)}.
\]

\section{Numerical Values for Higher Order Operators}
\label{ordertables}

In this appendix we provide tables of numerical values for the entries of higher order approximations of derivative operators, specifically the first derivative and the Laplacian. We also include a table of values for factorizing higher order Laplacians, and we discuss how to deal with factorizing stencil based Laplacians in more than one dimension. We use $k\th$ order to indicate that at lattice spacing $a$, the leading error term in the discrete derivative is of order $a^k$.

\subsection{First Derivative}

Below is a table of numerical values $a_j$ used for higher order approximations of the first-order derivative. For a 1D space with periodic boundary conditions, the radius-$N$ approximation is $ \sum_{j=-N}^{N} a_j S^j$ where $S$ represents a cyclic permutation of the vertices, i.e., $S_{i,j} = \delta_{i,j+1 \mod M}$ for $M>2N+1$.\\
\\
{\footnotesize
	\begin{tabular}{|l|l|c|}
		\hline
		operator & $\partial / \partial x$ \\
		\hline
		radius $N$ & order $k$ & entries $a_{-N}$ to $a_N$ \\
		\hline
		1 & 2 & -1/2, 0 ,1/2 \\
		2 & 4 & 1/12, -2/3, 0, 2/3, -1/12 \\
		3 & 6 & -1/60, 3/20, -3/4, 0, 3/4, -3/20, 1/60 \\
		4 & 8 & 1/280, -4/105, 1/5, -4/5, 0, 4/5, -1/5, 4/105, -1/280 \\
		5 & 10 & -1/1260, 5/504, -5/84, 5/21, -5/6, 0, 5/6, -5/21, 5/84, -5/504, 1/1260 \\
		\hline
\end{tabular}}\\
\\

\subsection{1-D Laplacians}

If we take the second derivative of the Lagrange interpolation formula (truncated at the $N$-th order), we arrive at Eqn. \ref{lagrangelaplacian}. Using this expression, we can find the coefficients $a_j$ which let us write the Laplacian under periodic boundary conditions as $L= \sum_{j=-N}^{N} a_j S^j$. Since the Laplacian is symmetric $a_j=a_{-j}$. In the table below we give the values for $a_j$ for the first 5 orders of truncation.
\\
\begin{tabular}{|l|l|c|}
	\hline
	operator & $\partial^2 / \partial x^2$ \\
	\hline
	radius $N$ & order $k$ & $a_{0}$ to $a_N$\\
	\hline
	1 & 2 & -2,1 \\
	2 & 4 & -5/2,4/3,-1/12  \\
	3 & 6 & -49/18,3/2,-3/20,1/90\\
	4 & 8 & -205/72,8/5,-1/5,8/315,-1/560\\
	5 & 10 & -5269/1800,5/3,-5/21,5/126,-5/1008,1/3150  \\
	\hline
\end{tabular}\\
\\
In order to implement our algorithm using any of the above Laplacians, we need to know its incidence matrix factorization. A simple procedure for doing this is the following:

\begin{enumerate}
	\item Generate the coefficients of the Laplacian operator using the Lagrange interpolation formula.
	\item With these coefficients, write the Laplacian for a 1-D grid with periodic boundary conditions in the form $ \sum_{j=-N}^{N} a_j S^j$. Note $a_j=a_{-j}$ since Laplacians are symmetric.
	\item Build an ansatz for the incidence matrix of the form $B = \sum_{j=1}^{N}b_j (I-S^j) $.
	\item Calculate $ BB^{\dagger}$.
	\item Solve $BB^{\dagger}= \sum_{j=-N}^{N} a_j S^j$ for the values $b_j$.
\end{enumerate}

We choose the ansatz $B = \sum_{j=1}^{N}b_j (I-S^j) $ instead of one like $\sum_{j=1}^{N}c_j S^j$ so that $BB^{\dagger}$ automatically has zero sum rows and columns like a Laplacian under periodic boundary conditions. The table below gives values for $b_j$ which lead to various higher order Laplacians.

\begin{tabular}{|l|c|}
	\hline
	radius $N$ & $b_{1}$ to $b_N$ \\
	\hline
	1  & 1\\
	2  &  1.1547, - (0.5774 $\pm$ 0.5)  \\
	3  & 1.2192, -0.1247, 0.0101 \\
	&  0.1247, -1.2192, 1.1046 \\
	4 & -0.0465, 1.1508, -1.2284, 0.1076 \\
	& 1.2540, -0.1552, 0.0209, -0.0016 \\
	& 0.0209, -0.1552, 1.2540, -1.1181 \\
	& 1.2284, -1.1508, 0.0465, -0.0166 \\
	5  & -0.0041, 0.0306, -0.1762, 1.2756, -1.1262 \\
	&  1.2756, -0.1762, 0.0306,-0.0041,0.0003 \\
	& 0.0289, 1.0626, -1.3223, 0.2195, -0.0131\\
	& 0.2195, -1.3223, 1.0626, 0.02891, 0.0243 \\
	\hline
\end{tabular}

\subsection{2-D Laplacians}
\label{stencil_incidence}

If we restrict to decomposing Laplacians into the form $L_{\mathrm{tot}}=L_x+L_y$ (treating the total Laplacian operator as the sum of the Laplacians in the $x$ and $y$ directions) then we can factor them simply by concatenating incidence matrices, as described in Subsection \ref{2dandbeyond}. These Laplacians are a restricted case since they approximate the second derivative at vertex $(x,y)$ using only the values of the function at vertices in the set $\{(x,y+r) | r \in \{-k, -k+1 \dots k-1,k \} \} \cup \{(x+r,y) | r \in \{-k, -k+1 \dots k-1,k \} \}$ (i.e. using vertices lying on a $+$-sign shaped subset of the vertices at distance $\leq r$ from $(x,y)$). 

Another well-known way to approximate Laplacians in multiple dimensions is to use \emph{stencils} such as the one in Figure \ref{stencil}. These have the disadvantage that their incidence matrices are not simply the concatenation of incidence matrices for Laplacians in the $x$ and $y$ directions; however, our procedure for calculating incidence matrix factorizations in this case can generalize. Using stencils has the advantage that they approximate the Laplacian at $(x,y)$ using all points within some distance $r$ of $(x,y)$ and not just those within distance $r$ in the $x$ of $y$ direction. 

We show how to factor the Laplacian corresponding to the stencil in Fig. \ref{stencil} which has error of order $a^2$. The formula is
\begin{equation}
\begin{array}{l} \frac{1}{a^{2}}\left\{ -\frac{2}{15}\left[\phi(x,y+2a)+\phi(x,y-2a)+\phi(x+2a,y)\right.\right.\vspace{5pt}  \\
\left.+\phi(x-2a,y)\right]-\frac{1}{10}\left[\phi(x+a,y+a)+\phi(x-a,y+a)\right.
\vspace{5pt} \\
\left.+\phi(x+a,y-a)+\phi(x-a,y-a)\right]+\frac{26}{15}\left[\phi(x+a,y)\right. \vspace{10pt} \\
\left.\left.+\phi(x-a,y)+\phi(x,y+a)+\phi(x,y-a)\right]-6\phi(x,y)\right\}  \vspace{10pt} \\
= \nabla^2 \phi(x,y) - \frac{a^2}{20} \left( \nabla^2 \right)^2 + O(a^6), \end{array}
\end{equation}
as one can verify by Taylor expansion. 
Previously we assumed we worked in a large one dimensional space with periodic boundary conditions; in this case we assume we're working on a large 2D space with periodic boundaries which can be treated as a torus discretized using a square grid. The Laplacian matrix can then be expressed as 
\begin{eqnarray}
L &=& -6I + \frac{26}{15} (S \otimes I + S^{\dagger} \otimes I + I \otimes S+ I \otimes S^{\dagger})\\
&-&\frac{1}{10} (S \otimes S + S \otimes S^{\dagger} + S^{\dagger} \otimes S + S^{\dagger} \otimes S^{\dagger})\nonumber\\
&-&\frac{2}{15} (S^2 \otimes I + (S^{\dagger})^2 \otimes I + I \otimes S^2+ I \otimes (S^{\dagger})^2)\nonumber
\end{eqnarray}

Our ansatz for the incidence matrix is

\[ B = \left[ \sum_{j,k}^{|j|+|k| \leq N}b_{j,k} (I-S^j \otimes S^k)  \text{ \Huge $|$ }   \sum_{j=-N}^{N}c_{j} (I-S^j \otimes I)   \right] \]

where $[A|B]$ denotes the horizontal concatenation of matrices $A$ and $B$. By construction this ansatz has zero-sum rows.

In terms of hypergraphs, this incidence matrix has hyperedges connecting vertices at distance at most $2N$ from each other, so the stencil they produce will have diameter at most $4N$. In fact there are two types of hyperedges present. Those encoded in the left block of the incidence matrix (the part where the coefficients $b_{j,k}$ appear) are hyperedges which span all $N$ neighbors of their center vertices; those encoded in the right block span all $N$ neighbors of their center vertex which have the same $y$ coordinate.

The stencil in Figure \ref{stencil} has diameter 4, and to factor it it suffices to set $N=1$. Doing so we find 16 solutions for the coefficients $b_{j,k}$ and $c_j$, one of which is

\begin{align}
b_{0,1} & = \frac{1}{46} \left(\frac{1}{5} \left(-\sqrt{345}-15\right)+3\right) \nonumber \\
b_{1,0} &= \frac{1}{30} \left(-\sqrt{345}-15\right) \nonumber \\
b_{-1,0} &= \frac{1}{30} \left(-\sqrt{345}-15\right)+1 \nonumber \\
b_{0,-1}&= \frac{1}{46} \left(\frac{1}{5} \left(-\sqrt{345}-15\right)+3\right) \nonumber \\
c_1 &= \frac{1}{138} \left(-2 \sqrt{1794}-69\right) \nonumber \\
c_{-1} &= \frac{1}{138} \left(-2 \sqrt{1794}-69\right)+1 
\end{align}

One might expect to find solutions with $c_j=0$ for all $j$; however, they don't exist. This reveals the importance of choosing the right ansatz for an incidence matrix factorization. For example, when factoring a 3D Laplacian built from a stencil with diameter $4N$, one might try the ansatz
\[ B' = \left[ \sum_{j,k,l}^{|j|+|k|+|l| \leq N}b_{j,k,l} (I-S^j \otimes S^k \otimes S^l)  \text{ \Huge $|$ }   \sum_{j=-N}^{N}c_{j} (I-S^j \otimes I \otimes I)   \right] \]

\noindent and not find solutions, while the ansatz
\begin{eqnarray*}
	B''&=&\left[\sum_{j,k,l}^{|j|+|k|+|l|\leq N}b_{j,k,l}(I-S^{j}\otimes S^{k}\otimes S^{l})\right.\\&\text{ \Huge\ensuremath{|} }&\left.\sum_{j=-N}^{N}c_{j}(I-S^{j}\otimes I\otimes I)\text{ \Huge\ensuremath{|} }\sum_{j=-N}^{N}d_{j}(I-I\otimes S^{j}\otimes I)\right]
\end{eqnarray*}

\noindent might have solutions.


\begin{thebibliography}{10}
	
	\bibitem{berry2014high}
	Dominic~W Berry.
	\newblock High-order quantum algorithm for solving linear differential
	equations.
	\newblock {\em Journal of Physics A: Mathematical and Theoretical},
	47(10):105301, 2014.
	
	\bibitem{berry2017quantum}
	Dominic~W. Berry, Andrew~M. Childs, Aaron Ostrander, and Guoming Wang.
	\newblock Quantum algorithm for linear differential equations with
	exponentially improved dependence on precision.
	\newblock {\em Communications in Mathematical Physics}, 356(3):1057--1081, Dec
	2017.
	
	\bibitem{leyton2008quantum}
	Sarah~K Leyton and Tobias~J Osborne.
	\newblock A quantum algorithm to solve nonlinear differential equations.
	\newblock {\em arXiv preprint arXiv:0812.4423}, 2008.
	
	\bibitem{cao2013quantum}
	Yudong Cao, Anargyros Papageorgiou, Iasonas Petras, Joseph Traub, and Sabre
	Kais.
	\newblock Quantum algorithm and circuit design solving the poisson equation.
	\newblock {\em New Journal of Physics}, 15(1):013021, 2013.
	
	\bibitem{clader2013preconditioned}
	B.~D. Clader, B.~C. Jacobs, and C.~R. Sprouse.
	\newblock Preconditioned quantum linear system algorithm.
	\newblock {\em Phys. Rev. Lett.}, 110:250504, Jun 2013.
	
	\bibitem{montanaro2016quantum}
	Ashley Montanaro and Sam Pallister.
	\newblock Quantum algorithms and the finite element method.
	\newblock {\em Phys. Rev. A}, 93:032324, Mar 2016.
	
	\bibitem{harrow2009quantum}
	Aram~W. Harrow, Avinatan Hassidim, and Seth Lloyd.
	\newblock Quantum algorithm for linear systems of equations.
	\newblock {\em Phys. Rev. Lett.}, 103:150502, Oct 2009.
	
	\bibitem{childs2015quantum}
	Andrew~M. Childs, Robin Kothari, and Rolando~D. Somma.
	\newblock Quantum algorithm for systems of linear equations with exponentially
	improved dependence on precision.
	\newblock {\em SIAM Journal on Computing}, 46(6):1920–1950, Jan 2017.
	
	\bibitem{berry2015hamiltonian}
	D.~W. Berry, A.~M. Childs, and R.~Kothari.
	\newblock Hamiltonian simulation with nearly optimal dependence on all
	parameters.
	\newblock In {\em 2015 IEEE 56th Annual Symposium on Foundations of Computer
		Science}, pages 792--809, Oct 2015.
	
	\bibitem{Grover2000}
	Lov~K. Grover.
	\newblock Synthesis of quantum superpositions by quantum computation.
	\newblock {\em Phys. Rev. Lett.}, 85:1334--1337, Aug 2000.
	
	\bibitem{Zalka}
	Christof Zalka.
	\newblock Efficient simulation of quantum systems by quantum computers.
	\newblock {\em Fortschritte der Physik}, 46:877--879, 1998.
	
	\bibitem{Grover_Rudolph}
	Lov Grover and Terry Rudolph.
	\newblock Creating superpositions that correspond to efficiently integrable
	probability distributions.
	\newblock {\em arXiv preprint quant-ph/0208112}, 2002.
	
	\bibitem{DP80}
	J.R. Dormand and P.J. Prince.
	\newblock A family of embedded runge-kutta formulae.
	\newblock {\em Journal of Computational and Applied Mathematics}, 6(1):19 --
	26, 1980.
	
	\bibitem{STAB}
	Lax P.~D. and Ichtmyer R.~D.~R.
	\newblock Survey of the stability of linear finite difference equations.
	\newblock {\em Communications on Pure and Applied Mathematics}, 9:267--293, May
	1956.
	
	\bibitem{colbert1992novel}
	Colbert Daniel~T. and Miller William~H.
	\newblock A novel discrete variable representation for quantum mechanical
	reactive scattering via the s-matrix kohn method.
	\newblock {\em The Journal of chemical physics}, 96:1982--1991, 1992.
	
	\bibitem{mexico}
	Choptuik Matthew~W.
	\newblock {\em Lectures for VII Mexican School on Gravitation and Mathematical
		Physics; Relativistic and Numerical Relativity; Numerical Analysis for
		Numerical Relativists}.
	\newblock University of British Columbia, 2009.
	
	\bibitem{PW60}
	Payne L.~E. and Weinberger H.~F.
	
	\bibitem{AC11}
	Andrews Ben and Clutterbuck Julie.
	\newblock Proof of the fundamental gap conjecture.
	\newblock {\em Journal of the Americal Mathematical Society}, 24:899--916, Mar
	2011.
	
	\bibitem{ambainis}
	Ambainis Andris.
	\newblock Variable time amplitude amplification and faster quantum algorithm
	for solving systems of linear equations.
	\newblock {\em arXiv:1010.4458 [quant-ph]}, 2012.
	
	\bibitem{BACS07}
	Dominic~W. Berry, Graeme Ahokas, Richard Cleve, and Barry~C. Sanders.
	\newblock Efficient quantum algorithms for simulating sparse hamiltonians.
	\newblock {\em Communications in Mathematical Physics}, 270(2):359--371, Mar
	2007.
	
\end{thebibliography}
\end{document}